\documentclass[11pt,a4paper]{article}
\usepackage{latexsym,amssymb,amsmath,amsthm,amsfonts,enumerate,verbatim,xspace,
exscale}

\input xy
\xyoption{all} \CompileMatrices \UseComputerModernTips

\usepackage{graphicx,tabularx}
\usepackage{amsmath,amsbsy,amsfonts,amssymb}
\usepackage{color}
\usepackage{amsthm}

\title{Hamiltonization of solids of revolution through reduction}

\author{P. Balseiro  \footnotemark}
\author{{\sc{Paula Balseiro}\thanks{
         Universidade Federal Fluminense, Instituto de Matem\'atica, Rua Mario Santos Braga S/N, 24020-140,
        Niteroi, Rio de Janeiro, Brazil. \newline{\texttt{E-mail: pbalseiro@vm.uff.br}}}
        }  \ \
}

\theoremstyle{plain}
\newtheorem{theorem}{Theorem}[section]
\newtheorem{lemma}[theorem]{Lemma}
\newtheorem{proposition}[theorem]{Proposition}
\newtheorem{corollary}[theorem]{Corollary}
\newtheorem*{theorem*}{Theorem}
\newtheorem{remarkth}[theorem]{Remark}

\theoremstyle{definition}
\newtheorem{definition}[theorem]{Definition}
\newtheorem{example}[theorem]{Example}

\newenvironment{remark}{\begin{remarkth}\upshape}{\hfill$\diamond$\end{remarkth}}

\newcommand{\g}{\mathfrak{g}}

\def\W{\mathcal{W}}
\def\M{\mathcal{M}}

\def\V{\mathcal{V}}
\def\S{\mathcal{S}}
\def\C{\mathcal{C}}
\def\Ham{\mathcal{H}}

\def\R{\mathbb{R}}

\def\red{{\mbox{\tiny{red}}}}
\def\nh{{\mbox{\tiny{nh}}}}

\def\B{{\mbox{\tiny{$B$}}}}
\def\subW{{\mbox{\tiny{$\W$}}}}
\def\subC{{\mbox{\tiny{$\C$}}}}

\def\subM{{\mbox{\tiny{$\M$}}}}
\def\subQ{{\mbox{\tiny{$Q$}}}}
\def\subSS{{\mbox{\tiny{$S$}}}}
\def\subWW{{\mbox{\tiny{$W$}}}}

\def\vecOm{\boldsymbol{\Omega}}

\def\a{\alpha}

\def\vecL{\boldsymbol{\lambda}}

\def\vecgamma{\boldsymbol{\gamma}}
\def\vecalpha{\boldsymbol{\alpha}}
\def\vecbeta{\boldsymbol{\beta}}
\def\veca{\boldsymbol{a}}

\let\OLDthebibliography\thebibliography
\renewcommand\thebibliography[1]{
  \OLDthebibliography{#1}
  \setlength{\parskip}{0pt}
  \setlength{\itemsep}{0pt plus 0.3ex}
}
\date{}
\begin{document}
\maketitle

\begin{abstract}
 In this paper we study the relation between conserved quantities of nonholonomic systems and the hamiltonization problem employing the geometric methods of \cite{Jac,PL2011}.  We illustrate the theory with classical examples describing the dynamics of solids of revolution rolling without sliding on a plane.  In these cases, using the existence of two conserved quantities we obtain, by means of {\it gauge transformations} and symmetry reduction, genuine Poisson brackets describing the reduced dynamics.
\end{abstract}

\tableofcontents

\section{Introduction}

Nonholonomic systems are mechanical systems with nonintegrable constraints in their velocities which, as a result, do not fit into the classical hamiltonian formalism. A central issue in the study of nonholonomic systems is determining 
whether they can be ``hamiltonized'' upon reduction by symmetries. This is known as the {\em hamiltonization} problem, about which there is a vast literature (see e.g. \cite{Bolsinov,BorisovMamaev2002, Chapligyn_reducing_multiplier, EhlersKoiller, FedorovJovan, Naranjo2008, JovaChap, Ohsawa} and references therein). 

This paper explores the connection between the presence of conserved quantities for a nonholonomic system and its hamiltonization, as raised in \cite{Naranjo2008}. Using the geometric techniques
developed in \cite{Jac,PL2011}, we show that, for certain types of symmetries admitting conserved quantities, one can distinguish particular 2-forms that can be used to modify the classical nonholonomic bracket (by means of {\em gauge transformations}); the reduction of such modified brackets to the orbit space are genuine Poisson brackets, relative to which the reduced equations of motion are hamiltonian. We show that all conditions for this procedure to work are met for a concrete set of examples, namely solids of revolution rolling on a plane without sliding as well as the classical example of an inhomogeneous ball rolling on a plane. 
As a consequence, we establish their hamiltonization, providing a geometric explanation for the reduced brackets found in \cite{BorisovMamaev2002,Ramos2004}.

Let us describe the mathematical set-up and results of the paper more precisely.

\noindent{\bf Nonholonomic systems and hamiltonization}

A nonholonomic system on a manifold $Q$ is defined by a lagrangian $L:TQ\to \R$ (of mechanical type) and a {\em nonintegrable} subbundle $D\subset TQ$ (the permitted velocities), see e.g. \cite{Blochbook,BorisovMamaev2008,Book:CDS}.
In this paper we will treat nonholonomic systems through their {\em hamiltonian} formalism: the lagrangian $L$ and the distribution $D$ induce a submanifold $\M \subset T^*Q$, an almost Poisson bracket $\{\cdot,\cdot\}_\nh$ and a hamiltonian function $\Ham_\subM :\M \to \R$, in such a way that the nonholonomic dynamics on $\M$ is determined by the vector field  
$$X_\nh = \{ \cdot, \Ham_\subM\}_\nh,$$
see e.g.  \cite{IbLeMaMa1999,Marle,SchaftMaschke1994}. The equations of motion $\dot{c}(t) = X_\nh(c(t))$ are not hamiltonian. In geometric terms, what is happening is that the bracket $\{\cdot, \cdot\}_\nh$ fails to satisfy the Jacobi identity due to the non-integrability of the distribution $D$.   In the presence of symmetries, 
the dynamics can be reduced to the quotient space $\M/G$, being defined by the vector field $X_\red$ obtained by the push-forward of $X_\nh$. As mentioned above,  the hamiltonization problem studies  whether the reduced equations of motion,
$$
\dot{\bar{c}}(t) = X_\red (\bar c(t)),
$$ 
are hamiltonian or not. Note that the push-forward of $\{ \cdot, \cdot\}_\nh$ to $\M/G$ defines a bracket $\{ \cdot, \cdot \}_\red$ which describes the reduced dynamics via $X_\red = \{ \cdot, \Ham_\red\}_\red$, where $\Ham_\red:\M/G\to \R$ is the reduced hamiltonian. Although $\{ \cdot, \cdot\}_\nh$ is not a Poisson bracket, it may be that 
 $\{ \cdot, \cdot \}_\red$ is, in which case we say that the system admits a {\it hamiltonization}.

There is, however, a more general set-up for hamiltonization. Notice that, even if $\{ \cdot, \cdot \}_\red$ is not a Poisson bracket, there might still exist other brackets $\{ \cdot, \cdot\}'$ on $\M/G$, satisfying
\begin{equation} \label{Intro:1} 
X_\red = \{ \cdot, \Ham_\red\}',
 \end{equation}
which are Poisson. Following \cite{PL2011,Naranjo2008},
one way to find new brackets on $\M/G$ 
is to first consider new invariant brackets on $\M$. Those can be obtained through modifications of $\{ \cdot, \cdot \}_\nh$ by
gauge transformations \cite{SW}  by (invariant) 2-forms $B$. One then considers their reductions $\{ \cdot, \cdot\}_\red^\B$ on $\M/G$, and searches for Poisson brackets satisfying \eqref{Intro:1} within this family. In this way, the hamiltonization problem 
is translated into the search of 2-forms $B$ with suitable properties.

\noindent{\bf Results: hamiltonization and conserved quantities} 
 
Let us consider a nonholonomic system defined on a manifold $Q$, with symmetry group $G$. Following  \cite{Jac}, we will assume that the symmetries satisfy an additional property, called  {\it vertical-symmetry condition} (see Def.~\ref{Def:VertSym} below). Motivated by examples, we will assume that $G$ acts properly, though not necessarily freely. So we will work in the context of {\em singular reduction}, as in \cite{Book:BC,BS2016}. In practice, this means that we will formulate our results in terms of the ring of $G$-invariant functions on $\M$ (thought of as the ring of smooth functions on $\M/G$, viewed as a differential space). 
Following \cite{PL2011}, the main new aspect of the present work is that we will relate the 2-forms $B$ used to
gauge transform the nonholonomic bracket with the presence of first integrals of the dynamics. 

More precisely, let us assume the existence of a conserved quantity $J\in C^\infty(\M)$ which is a {\it horizontal gauge momentum} \cite{FSG2008,FSG2009} (see Def.~\ref{Def:Gauge-Momentum}). 
Contrary to what occurs in hamiltonian mechanics, the vector field $X_J= \{ \cdot, J\}_\nh$ may not be vertical (i.e., tangent to a $G$-orbit). We hence search for (invariant) 2-forms $B$ with the property that the modified nonholonomic bracket $\{\cdot,\cdot\}_\B$ is such that the vector field
$$
X_J^\B := \{ \cdot, J\}_\B
$$ 
is vertical with respect to the $G$-action (Thm.~\ref{T:Casimirs}).  If this holds, we show that the gauge transformed bracket $\{\cdot, \cdot\}_\B$ induces a reduced bracket $\{\cdot, \cdot\}_\red^\B$ on $\M/G$ for which $J$ is a Casimir (provided $J$ is also $G$-invariant).  We will observe that in various examples these Casimirs play a fundamental role in verifying that $\{\cdot, \cdot\}_\red^\B$ is a Poisson bracket. Concerning the existence of horizontal gauge momenta, we 
use the geometric framework of  \cite{Jac} to derive a ``momentum equation'' (see Prop.~\ref{P:gaugeMomentum}) in the spirit of the one
in \cite{BGMConservation}, but with a more clear dependence on the geometric information.  
 
We apply this theory to study the nonholonomic dynamics of a solid of revolution rolling without slipping on a plane,
which  includes the Routh sphere and the rolling (axisymmetric) ellipsoid \cite{Bolsinov,BorisovMamaev2002,Cushman1998, Book:CDS}. Following \cite{Jac}, we express the failure of the Jacobi identity of the nonholonomic bracket before and after reduction, which is controlled by the 2-form $\langle {\mathcal J},{\mathcal K}_\subW \rangle$ introduced in \cite{Jac}
(see Lemma~\ref{L:JK} and Prop.~\ref{P:notPoisson}). By using the momentum equation of Prop.~\ref{P:gaugeMomentum}, we derive a system of differential equations that leads us to an alternative way to express two (known) $G$-invariant horizontal gauge momenta $J_1$ and $J_2$, as in \cite{Book:CDS}. By analyzing how far the vector fields $X_{J_1}$ and $X_{J_2}$ are from being vertical with respect to the $G$-action, we devise a 2-form $B$ on $\M$ that is compatible with the dynamics, in the sense that $X_\nh = \{\cdot , \Ham_\subM \}_\B$. More importantly, this gauge-transformed bracket has the property that the vector fields $X_{J_i}^\B = \{\cdot , J_i \}_\B$ are vertical. As a consequence, the reduced bracket  $\{\cdot, \cdot\}_\red^\B$ on the (differential) space $\M/G$ (which satisfies \eqref{Intro:1} by construction) admits two Casimir functions, defined by $J_1$ and $J_2$ (see Thm.~\ref{T:Solids-Casimirs}); using this fact, one can directly verify that $\{\cdot, \cdot\}_\red^\B$ is a Poisson
bracket.  Motivated by \cite{Bolsinov}, we remark that our hamiltonization of rolling solids of revolution is relative to
the action of $G= E(2)\times S^1$,  and that we do not have analogous results using a smaller group of symmetries (see Remarks \ref{R:Chaplygin} and \ref{R:Bolsinov}).

\medskip

\noindent {\it Acknowledgments}: I thank CNPq (Brazil) for supporting this project.  I am grateful to Richard Cushman, Jedrzej Sniatycki, Larry Bates, Nicola Sansonetto and Alessia Mandini for stimulating conversations. I thank Dmitry Zenkov for the invitation to the CMS meeting in Edmonton in July 2016, where part of this work was presented.   I am especially indebted to  Luis Garcia-Naranjo for inspiring discussions, particularly concerning the symmetry group for the Routh sphere (Sec.~\ref{S:Routh}) ; his joint work with J. Montaldi \cite{LuisJames} contains results related to ours, but independently obtained, where the horizontal gauge momenta become Casimirs of an alternative reduced bracket.

\section{Nonholonomic mechanics: hamiltonization and conserved quantities}

\subsection{Preliminaries: Nonholonomic systems}

Consider a nonholonomic system on a manifold $Q$ defined by a lagrangian $L:TQ\to \R$ of mechanical type and a (non-integrable and constant rank) distribution $D$ on $Q$.  The distribution $D$ describes the permitted velocities of the system. Our first goal is to write the equations of motion of the system on the cotangent bundle using an almost Poisson bracket (see e.g., \cite{Blochbook,CdLMD,IbLeMaMa1999,MarsdenKoon}).  

Denoting by $\kappa$ the kinetic energy metric, we define the {\it constraint submanifold} $\M$ of $T^*Q$ by $\M:=\kappa^\flat(D)$, where $\kappa^\flat: TQ \to T^*Q$ is defined by $\kappa^\flat(X)(Y)=\kappa(X, Y)$ for $X,Y\in TQ$.  Since $\kappa$ is linear on the fibers, $\M$ is a vector subbundle of $T^*Q$; we denote by $\tau:\M \to Q$ the canonical projection. 

Let $\C$ be the non-integrable and constant rank distribution on $\M$ given, at each $m\in \M$, by 
\begin{equation}\label{Def:C}
\C_m :=\{ v_m \in T_m\M \ : \ T\tau(v_m) \in D_{\tau(m)} \}.
\end{equation} 

The lagrangian $L$ induces the hamiltonian function  $\Ham:T^*Q \to \R$.  Let us denote by $\Ham_\subM : \M \to \R$ the restriction of $\Ham$ to the submanifold $\M$, i.e., $\Ham_\subM := \iota^*\Ham$ where $\iota : \M \to T^*Q$ is the natural inclusion.  Let $\Omega_\subM$ be the 2-form on $\M$ given by $\Omega_\subM := \iota^*\Omega_{\subQ}$ where $\Omega_\subQ$ is the canonical 2-form on $T^*Q$.
Following \cite{BS93}, the nonholonomic dynamics is described by the integral curves of the vector field $X_\nh$ on $\M$ defined by 
\begin{equation}\label{Eq:Dyn}
{\bf i}_{X_\nh} \Omega_\subM |_\C = d\Ham_\subM |_\C, 
\end{equation}
where $|_\C$ denotes the point-wise restriction to $\C$. Since the vector field $X_\nh$ takes values on $\C$, we say that it is  a section of the bundle $\C\to\M$, i.e., $X_\nh\in \Gamma(\C)$.  It is important to note that the solution $X_\nh$ satisfying \eqref{Eq:Dyn} is unique since the 2-section $\Omega_\subM |_\C$ is nondegenerate \cite{BS93}.

The {\it nonholonomic bracket} $\{\cdot, \cdot \}_\nh$ on $C^\infty(\M)$ is given, for $f,g \in C^\infty(\M)$, by $\{f,g\}_\nh = -X_f(g)$, where $X_f \in \mathfrak{X}(\M)$ is the unique vector field such that
\begin{equation}\label{Eq:HamVectorField}
{\bf i}_{X_f} \Omega_\subM |_\C = df |_\C.
\end{equation}
The nonholonomic bracket was defined in \cite{CdLMD,Marle,SchaftMaschke1994} and shown to be an almost Poisson bracket: it is $\R$-bilinear, skew-symmetric and satisfies the Leibniz identity.  

We denote by $\pi_\nh$ the bivector field on $\M$ associated to $\{ \cdot, \cdot \}_\nh$, i.e., for $\alpha\in\Omega^1(\M)$ then $\pi_\nh^\sharp(\alpha) = - X$ if and only if ${\bf i}_{X} \Omega_\subM |_\C = \alpha |_\C,$. In other words, $\pi_\nh(df,dg) = \{f,g\}_\nh$ for $f,g\in C^\infty(\M)$. In these terms, the dynamics is described by the vector field $X_\nh$ given by $X_\nh = - \pi_\nh^\sharp(d\Ham_\subM)$.

\begin{remark}
 It is straightforward to see that $\{\cdot,\cdot\}_\nh$ fails to satisfy the Jacobi identity since its {\it characteristic distribution} --the distribution generated by the ``hamiltonian'' vector fields $\pi_\nh^\sharp(df)$ for $f\in C^\infty(\M)$-- is the distribution $\C$ defined in \eqref{Def:C}, which is not integrable. 
\end{remark}

Throughout this article we will denote by $(\M,\pi_\nh, \Ham_\subM)$ a nonholonomic system on the manifold $Q$ given by a lagrangian $L$ and a constraint distribution $D$.

\subsection{Symmetries and reduction}\label{Ss:Symmetries}

Let $G$ be a Lie group acting properly on the manifold $Q$.  We say that the $G$-action is a {\it symmetry} of the nonholonomic system (defined on $Q$ by a lagrangian $L:TQ\to \R$ and a distribution $D$) if the tangent lift of the action on $TQ$ leaves $L$ and $D$ invariant.  In this case, the cotangent lift of the action to $T^*Q$ leaves the submanifold $\M$ invariant, so we have a proper $G$-action on $\M$: 
$$
\phi: G \times \M \to \M.
$$
It follows that the hamiltonian $\Ham_\subM$ and the nonholonomic bracket $\{\cdot, \cdot \}_\nh$ are $G$-invariant.   Our next goal is to write the equations of motion in the quotient space $\bar\M:=\M/G$.

Consider the quotient map $\rho:\M\to\bar\M$ and endow $\bar\M$ with the quotient topology. Following \cite[Sec.~3]{Book:BC}, since the action is proper, we will view $\bar\M$ as a {\it differential space}, with ring of smooth functions $C^\infty(\bar\M)$ given by the $G$-invariant smooth functions on $\M$.

The $G$-invariant bracket  $\{\cdot, \cdot \}_\nh$ on $C^\infty(\M)$ induces an almost Poisson bracket on $C^\infty(\bar\M)$ denoted by $\{\cdot, \cdot \}_\red$. That is, for $f,g\in C^\infty(\bar\M)$, 
\begin{equation}\label{Eq:ReducedBracket}
\{f,g\}_\red \circ \rho = \{\rho^*f,\rho^*g\}_\nh.
\end{equation}

Since $X_\nh \in \mathfrak{X}(\M)$ satisfies that $T\phi_g X_\nh - X_\nh$ is tangent to the orbits of the $G$-action on $\M$, then $X_\nh$ descends to a vector field
$X_\red$ on $\bar\M$ (see \cite{BS2016}).  The reduced dynamics is described by the integral curves of the vector field $X_\red$ on $\bar\M$ given by
$$
\{ \cdot, \Ham_\red\}_\red = X_\red,
$$
where $\Ham_\red\in  C^\infty(\bar\M)$ such that $\rho^*\Ham_\red = \Ham_\subM$.

\bigskip

\subsection{Splitting adapted to the constraints and the vertical-symmetry case}

In order to study the failure of the Jacobi identity of $\{\cdot, \cdot \}_\red$ on $\bar\M$ (and afterwards to find a Poisson bracket in $\bar\M$),  in what follows, we will define a complement $W$ of the constraints $D$ in $TQ$ taking into account the symmetries. A complement $W$ was already defined and studied in \cite{Jac} for a free and proper action.  In our current situation (where we have a proper action) it takes a little more work to guarantee the existence of a smooth and constant rank complement.  

Consider a nonholonomic system given by the lagrangian $L$ of mechanical type and a (nonintegrable) distribution $D$ with a $G$-symmetry induced by a proper action of the Lie group $G$ on $Q$. %Recall that $V$ is the vertical space associated to the $G$-action on $Q$.% and on $\M$ respectively. 
Let us denote by $V$ the (generalized) distribution on $Q$ whose fibers $V_q$ are the tangent spaces to the orbits of $G$ in $Q$, that is $V_q = T_q(Orb_{G}(q))$ (as a consequence of the non-freeness of the $G$-action on $Q$, the distribution $V$ may vary its rank). 
Let $\mathfrak{g}$ be the Lie algebra associated to the Lie group $G$ and denote by 
\begin{equation}\label{Eq:g-action}
\Psi: \mathfrak{g} \to \mathfrak{X}(Q)
\end{equation}
the Lie algebra homomorphism such that $\Psi_q : \mathfrak{g} \to V_q\subset T_qQ$ is  $\Psi_q (\eta) = \eta_\subQ(q)$, where $\eta_\subQ(q)$ is the infinitesimal generator associated to $\eta\in\mathfrak{g}$ at $q\in Q$.  Observe that there might be a $q\in Q$ for which $\Psi_q: \mathfrak{g} \to V_q$ has a non trivial kernel, that is, the rank of $V$ may vary  as a consequence of the non-freeness of the action. 

We say that the $G$-symmetry verifies the {\it dimension assumption} (\cite{BKMM}) if 
\begin{equation} \label{Eq:DimAssumptionTQ}
 T_qQ = D_q +V_q \qquad \mbox{for each } q \in Q.
\end{equation}
%or equivalently if $T_m\M = \C_m + \V_m$ where $m \in \M$. 
% \item[$(ii)$] and the {\it regularity assumption} if the rank of the distribution $V/S$ on $Q$ has constant rank for $S$ the distribution defined in \eqref{Def:S}. 
%  \end{enumerate}
% \end{definition}
Let $S$ be the (generalized) distribution on $Q$ given, at each $q\in Q$, by 
\begin{equation}\label{Def:S}
S_q := D_q \cap V_q. %\qquad \mbox{and} \qquad \S_m = \C_m\cap\V_m,
\end{equation}
% and $m\in \M$. 

\begin{proposition}
 The dimension assumption guarantees the existence of a constant rank smooth distribution $W$ on $Q$ such that $W_q \subset V_q$  for all $q\in Q$ and 
\begin{equation}\label{Eq:SplittingS}
V_q = S_q \oplus W_q.
\end{equation}
\end{proposition}

\begin{proof} 
%Observe that the dimension assumption asserts that the map $\Psi:\g \to TQ$ (defined as in \eqref{Eq:g-action}) is transverse to $D$, %i.e., $TQ = \textup{Im}\Psi + D$ and hence 
Consider the vector bundle $(\g\times Q)\oplus D \to Q$ and let us define the map $\psi: (\g\times Q ) \oplus D \to TQ$  given by $\psi(\xi,v) = \Psi(\xi) -v$, where $\Psi:\g\times Q \to TQ$ is the vector bundle map defined by \eqref{Eq:g-action}. Observe that the dimension assumption ensures that $\textup{Im}\psi= TQ$, so  $\textup{Im}\psi$ has constant rank.  As a consequence, 
$$
\textup{Ker}\psi =\{(\xi,v) \ : \ \Psi(\xi) = v\}\subset(\g\times Q)\oplus D
$$ 
is a subbundle.  By projecting on the first factor, we see that $ \mathfrak{g}_\subSS:= \{ \xi\in \mathfrak{g}\times Q \ : \ \Psi(\xi) \in S\}$
 is a subbundle of $\g\times Q\to Q$.

Let us now choose any subbundle $\g_\subWW \to Q$ of  $\g\times Q \to Q$ such that 
$$
\g \times Q = \g_\subSS \oplus \g_\subWW,
$$
and define $W:=\Psi(\g_\subWW)$. 
Since $\textup{Ker}\, \Psi \subseteq \g_\subSS$, we see that $\Psi|_{\g_\subWW}$ is an isomorphism onto $W$, so $W$ is a subbundle of $TQ$. Moreover, $V = S\oplus W$.
% 
% 
% 
% 
% First, we prove that $V/S\to Q$ is a vector bundle.  In fact, by the dimension assumption \eqref{Eq:DimAssumptionTQ} and using that $\textup{rank}\,D$ is constant, we see that $\textup{rank}\,(V/S)$ is constant too since $\textup{rank}\, TQ = \textup{rank}\,D +\textup{rank}\,V - \textup{rank}\,S = \textup{rank}\,D + \textup{rank}\, (V/S)$.  Moreover, the distribution $V/S$ is smooth since it is the image of the smooth map $p:TQ\to TQ/D$ restricted to the smooth (generalized) distribution $V$, i.e., $p(V) = V/S$. Hence, $V/S\to Q$ a vector bundle since it is a constant rank smooth distribution. 
% 
% Therefore, we choose $W\subset V$ such that $W_q$ is isomorphic to $(V/S)_q$.  For example, using the $G$-invariant kinetic energy $\kappa$, we may define $W_q = S^\perp_q \cap V_q$. 

\end{proof}

In summary, the splitting $V=S\oplus W$ is induced from a vector-bundle
splitting 
\begin{equation} \label{Eq:gS+gW}
(\g \times Q) = \mathfrak{g}_\subSS  \oplus \mathfrak{g}_\subWW,
\end{equation}
where $\mathfrak{g}_\subSS \to Q$ is the vector bundle 
\begin{equation} \label{Def:gS}
 {\mathfrak{g}_\subSS} = \{ (\xi,q)  \in \mathfrak{g}\times Q \ : \ \xi_\subQ(q) \in S_q\},
 \end{equation}
and $W:=\Psi(\g_\subWW)$.
As we saw, even if $S$ varies its rank, $\mathfrak{g}_\subSS\to Q$ has constant rank.

Note that there might be many subbundles $W\subset TQ$ satisfying $V= S\oplus W$.
Following \cite{Jac}, we call such $W$ a {\it vertical complement of the constraints} $D$, since
\begin{equation}\label{Eq:VertComplementTQ} 
TQ = D \oplus W, \qquad \mbox{and} \qquad W \subseteq V.
\end{equation}

Throughout this paper, {\bf we will assume that the nonholonomic system has a $G$-symmetry satisfying the dimension assumption}, so that it is possible to choose a vertical complement of the constraints.

\begin{definition}\label{Def:VertSym}
We say that a vertical complement of the constraints $W$ has the {\it vertical-symmetry condition} if there is a subspace  $\mathfrak{w}$ of $\mathfrak{g}$ such that $\mathrm{Ad}_g(\mathfrak{w})\subseteq \mathfrak{w}$ for all $g\in G$, and $\Psi:\g\times Q \to TQ$ restricts to an isomorphism $(\mathfrak{w}\times Q)\cong W$.

%given 
%at each $q\in Q$,  by 
%$$
%\Psi_q|_{\mathfrak{w}}: \mathfrak{w} \to W_q,
%$$ 
%is an isomorphism.  
\end{definition}

In this case, the trivial bundle $\g_\subWW = \mathfrak{w} \times Q \to Q$ is a complement of $\g_\subSS\to Q$ in $\g \times Q \to Q$ (note that $\g_\subSS\to Q$ is also a trivial bundle).
Moreover, if $W=\Psi( \mathfrak{w} \times Q)$ for a subspace  $\mathfrak{w}\subseteq \g$, then the $\mathrm{Ad}$-invariance of $\mathfrak{w}$ is equivalent to the $G$-invariance of $W$.

%It is worth noticing that if the vertical complement $W$ has the {\it vertical-symmetry condition} (and $G$ is connected) then it will be $G$-invariant (because of the $Ad$-invariance property of $\mathfrak{w})$) and we can talk about a $G$-invariant complement of the constraints.  %Moreover,  looking for a complement $W$ satisfying the {\it vertical-symmetry condition} is equivalent to look for a $Ad$-invariant subalgebra $\mathfrak{w}$ such that $\mathfrak{w} \times Q$ complements $\g_\subSS$ in $\g\times Q$. 

\begin{remark}
 \begin{enumerate}
 
 \item[$(i)$] The vertical-symmetry condition is equivalent to asking for the existence of a normal subgroup $G_\subWW$ of $G$ such that the action of $G_\subWW$ on $Q$ is locally free and that the associated vertical space is $W$ (i.e., $W_q=T_q(Orb_{G_\W}(q))$\,).

 \item[$(ii)$] Choosing a vertical complement $W$ with the vertical-symmetry condition often simplifies the theory, as we showed in \cite{Jac,PO2015}.  In this article, we study examples that admit such a complement that lead to a simpler interpretation of their geometry. 
  
 %\item[$(iii)$]  If the $G$-action on $Q$ is free and proper, there is always a $G$-invariant complement of the constraints given by $W=S^\perp \cap V$ where $S^\perp$ denotes the orthogonal complement of $S$ with respect to the kinetic energy metric.  For a $G$-invariant $W$, the vertical-symmetry condition requires only the existence of the subalgebra $\mathfrak{w}$ such that $W$ is generated by the infinitesimal generators (see \cite{Jac,PO2015}), since the $G$-invariance of $W$ implies the $Ad$-invariance of  $\mathfrak{w}$ (and hence, $\mathfrak{w}$ is an ideal in $\g$).
 \end{enumerate}
\end{remark}
% 
% Let us suppose that the vertical complement $W$ in the splitting \eqref{Eq:VertComplementTQ} has the vertical-symmetry condition.  
% Then \eqref{Eq:SplittingS} induces a splitting of the trivial Lie algebra bundle $ \mathfrak{g}\times Q\to Q$ so that, at each point $q\in Q$, $\mathfrak{g}|_q =  \Psi_q^{-1} (S_q)  \oplus \mathfrak{w} |_q$  (since by Def.~\ref{Def:VertSym}, $\textup{Ker}\, \Psi_q \subset \Psi_q^{-1} (S_q)$). 

\bigskip

We will now lift the splitting \eqref{Eq:VertComplementTQ} to $T\M$.
First, let $\V$ be the vertical space associated to the $G$-action on $\M$, that is, at each $m\in\M$, $\V_m = T_m(Orb_G(m))$.  
The subalgebra $\mathfrak{w}$ induces the ($G$-invariant and constant rank) distribution $\W$ on $\M$ given, at each $m\in \M$, by 
\begin{equation}\label{Eq:WonTM}
\W_m := \textup{span}\{ \eta_\subM(m)\ : \ \eta \in \mathfrak{w} \},
\end{equation}
where $\eta_\subM$ is the infinitesimal generator of the element $\eta\in \mathfrak{w}$.
It is clear that $\W$ is a {\it vertical complement of the constraint distribution} $\C$ satisfying the vertical-symmetry condition: 
\begin{equation}\label{Eq:VertComplementTM}
T_m\M = \C_m \oplus \W_m \qquad \mbox{for} \quad \W_m \subset \V_m,  
\end{equation}
for each $m\in \M$. Therefore, if we denote by $\S$ the distribution on $\M$ given by $\S_m := \C_m \cap \W_m$, then 
\begin{equation}\label{Eq:SplittingSS}
\V_m = \S_m \oplus \W_m.
\end{equation}
Observe that $\W_m$ and $W_{\tau(m)}$ are isomorphic through the canonical projection $\tau:\M\to Q$. However $\S_m$ and $S_{\tau(m)}$ (or $\V_m$ and $V_{\tau(m)}$) are not necessarily isomorphic since there might be a $\bar{q}\in Q$ and $\xi\in\Gamma(\g_{\subSS})$ such that $\xi_\subQ(\bar{q}) = 0$ but $\xi_\subM(\bar{m}) \neq 0$ for $\bar{m}\in \M$ such that $\tau(\bar{m})=\bar{q}$.

\subsection{The 2-form $\langle {\mathcal J}, {\mathcal K}_\subW\rangle$ and the failure of the Jacobi identity}

Following \cite{Jac}, the vertical complement $\W$ given in \eqref{Eq:WonTM} induces a 2-form $\langle{\mathcal J}, {\mathcal K}_\subW\rangle$ on $\M$ (see Def.~\ref{Def:JK} below) that will characterize the failure of the almost Poisson brackets $\{\cdot, \cdot \}_\nh$ and $\{\cdot, \cdot\}_\red$.

Consider a nonholonomic system $(\M,\pi_\nh, \Ham_\subM)$ with a $G$-symmetry. Suppose that the nonholonomic system verifies the dimension assumption and pick a vertical complement $W$ of the constraints $D$. From now on, we assume that the complement $W$ satisfies the vertical-symmetry condition (Def.~\ref{Def:VertSym}).  

Let $\W$ be the $G$-invariant vertical complement on $\M$ defined in \eqref{Eq:WonTM} and we denote by $P_\subC:T\M \to \C$ and $P_\subW:T\M \to \W$ the projections associated to decomposition \eqref{Eq:VertComplementTM}. Then we consider the map ${\mathcal A}_\subW: T\M \to \mathfrak{g}$ defined by 
$$
{\mathcal A}_\subW (v_m) = \xi(m) \quad \mbox{if and only if} \quad P_\subW (v_m) = \xi_\subM(m),
$$
for each $m\in\M$.
The $\W${\it-curvature} is the $\mathfrak{g}$-valued 2-form on $\M$ given by
\begin{equation} \label{Def:KinCurv-2form}
\mathcal{K}_\subW(X,Y):=  d {\mathcal A}_\subW(P_\subC (X),P_\subC (Y)) \qquad \mbox{ for } X,Y \in T\M.
\end{equation}
On the other hand, for $\iota:\M\to T^*Q$ the inclusion, we denote by
\begin{equation} \label{Def:Momentum}
\mathcal{J} := \iota^*J_\subQ : \M\to \mathfrak{g}^*
\end{equation}
the restriction to $\M$ of the canonical momentum map $J_\subQ:T^*Q\to \mathfrak{g}^*$.

\begin{definition}[\cite{Jac}] \label{Def:JK}
The 2-form $\langle \mathcal{J}, \mathcal{K}_\subW \rangle$ on $\M$ is defined by the natural pairing between the function $\mathcal{J}:\M\to \mathfrak{g}^*$ and the $\mathfrak{g}$-valued 2-form $\mathcal{K}_\subW$ defined in \eqref{Def:KinCurv-2form} and \eqref{Def:Momentum} respectively. 
\end{definition}

Recall that the $G$-action on $Q$ (and on $\M$) is proper but not necessarily free.  However, the vertical-symmetry condition enables the definition of a $G$-invariant vertical complement $\W$ that is isomorphic to $W$ and thus the definition and properties of the $\W$-curvature $\mathcal{K}_\subW$ follow from \cite{Jac}. In particular, we obtain that the 2-form $\langle \mathcal{J}, \mathcal{K}_\subW \rangle$ is $G$-invariant.

As it has been studied in \cite{Jac}, the 2-form  $\langle \mathcal{J}, \mathcal{K}_\subW \rangle$ encodes the failure of the Jacobi identity of the nonholonomic bracket $\pi_\nh$: 
\begin{equation}\label{Eq:Jac-pi_nh} 
\tfrac{1}{2}[\pi_\nh , \pi_\nh]  = - \pi_\nh^\sharp ( d\langle {\mathcal J}, \mathcal{K}_\subW \rangle ) - \Psi_\nh,
\end{equation}
where $[ \cdot , \cdot ]$ is the Schouten bracket\footnote{If $\pi$ is a bivector field on $\M$, the Schouten bracket $[\pi,\pi]$ is a 3-vector field such that, for $f,g,h\in C^\infty(\M)$,  $\tfrac{1}{2}[\pi,\pi](df,dg,dh) = cyclic [\, \{f,\{g,h\}\}\, ] $, where $\{\cdot,\cdot\}$ is the bracket associated to $\pi$ (and where ``$cyclic [\, \cdot \, ]$'' denotes the cyclic sum), see \cite{Book:MR}.} and $\Psi_\nh$ is a 3-vector field given, for 1-forms $\alpha, \beta,\gamma$ on $\M$, by 
$$
\Psi_\nh(\alpha, \beta, \gamma) = cyclic \left[ \gamma\left(\left(\mathcal{K}_\subW( \pi_\nh^\sharp(\alpha), \pi_\nh^\sharp(\beta) ) \right)_\subM \right) \right].
$$
(Recall also that if $\Phi$ is a 3-form on $\M$, then for $\alpha,\beta,\gamma$ 1-forms on $\M$, we have $\pi_\nh^\sharp(\Phi)(\alpha,\beta,\gamma) = - \Phi (\pi_\nh^\sharp(\alpha), \pi_\nh^\sharp(\beta), \pi_\nh^\sharp(\gamma))$).

From Section \ref{Ss:Symmetries}, the reduced bracket $\{\cdot, \cdot\}_\red$ on  $C^\infty(\bar\M)$ defined in \eqref{Eq:ReducedBracket} is an almost Poisson bracket.
However, we can also use \cite{Jac} to characterize the failure of the Jacobi identity of $\{\cdot, \cdot\}_\red$ on $\bar\M$ even when $\bar\M$ is a differential space.

\begin{proposition}\label{P:Red-Jac} If the proper $G$-action on $\M$ satisfies the dimension assumption and the vertical complement $W$ satisfies the vertical-symmetry condition then, for $f,g,h$ smooth functions on $\bar\M$, we have  
\begin{equation}\label{Eq:Jac-pi_red}
cyclic \left[ \, \{f,\{g,h\}_{\emph\red}\}_{\emph\red} \circ \rho \, \right]   =  d\langle {\mathcal J}, \mathcal{K}_\subW \rangle ( \pi_{\emph\nh}^\sharp(d\rho^*f), \pi_{\emph\nh}^\sharp(d\rho^*g),\pi_{\emph\nh}^\sharp(d\rho^*h) ).
\end{equation}
\end{proposition}

Therefore, the 3-form $d\langle {\mathcal J}, \mathcal{K}_\subW \rangle$ tells us if $\{\cdot, \cdot\}_\red$ is Poisson, or not:  If the right hand side of \eqref{Eq:Jac-pi_red} is zero for all $f,g,h\in C^\infty(\bar\M)$, then $\{\cdot, \cdot\}_\red$ is Poisson.

\subsection{Symmetries and conserved quantities}\label{Ss:Symm-Conservation}

In this section we study conserved quantities of the nonholonomic system $(\M,\pi_\nh, \Ham_\subM)$ that appear as a consequence of the presence of symmetries.

The {\it nonholonomic momentum map} $J^{\nh} : \M \to \mathfrak{g}_\subSS^*$ is given, for $m \in \M$ and $\xi \in \Gamma(\mathfrak{g}_\subSS)$, by
$$
\langle J^{\nh}(m) , \xi_{\tau(m)} \rangle = {\bf i}_{\xi_\M} \Theta_\subM (m),
$$
where $\Theta_\subM$ is the pullback to $\M$ of the canonical 1-form $\Theta_\subQ$ on $T^*Q$: $\Theta_\subM := \iota^*\Theta_\subQ$ (see \cite{BKMM}) and where $\xi_\subM (m) := (\xi_{\tau(m)})_\subM(m)$.

%If there exists a constant section $\eta \in \mathfrak{g}_\subSS$ then $\langle J^{\nh} , \eta \rangle = \langle {\mathcal J} , \eta \rangle $ for ${\mathcal J}: \M \to \mathfrak{g}^*$ defined in \eqref{Def:Momentum}.
            
%Observe also that the definition of the nonholonomic momentum map does not involve the chosen vertical complement $W$.

\begin{definition}[\cite{FSG2008}]\label{Def:Gauge-Momentum}
A {\it horizontal gauge momentum} of the nonholonomic vector field $X_\nh$ is a function $J \in C^\infty(\M)$ for which there exists $\chi\in\Gamma(\g_\subSS)$ such that 
$$
J= \langle J^\nh , \xi \rangle = {\bf i}_{\chi_\subM}\Theta_\subM, 
$$
and in such a way that $X_{\nh}(J) = 0$. The section $\chi\in \Gamma(\g_\subSS)$ is called the {\it horizontal gauge symmetry}. 
\end{definition}

As it was studied in many references \cite{BGMConservation, BKMM, Crampin-Mestdag, FSG2009, FSG2008, FS2015,  Jovanovic, Marle95, sniatycki, Zenkov}, the problem in finding a horizontal gauge momentum $J$ of $X_\nh$ resides in finding the associated horizontal gauge symmetry $\xi$.  In order to study properties of a gauge momentum and its associated gauge symmetry we recall the definition of the 2-form $\Omega_{\mbox{\tiny{$\mathcal{J}\!\mathcal{K}$}}}$ on $\M$ and its relation with the nonholonomic momentum map given in \cite{Jac}.

Given a ($G$-invariant) vertical complement of the constraints $W$ as in \eqref{Eq:VertComplementTQ}, we define the ($G$-invariant) 2-form $\Omega_{\mbox{\tiny{$\mathcal{J}\!\mathcal{K}$}}}$ on $\M$ by 
\begin{equation} \label{Eq:OmegaJK}
 \Omega_{\mbox{\tiny{$\mathcal{J}\!\mathcal{K}$}}}:= \Omega_\subM + \langle {\mathcal J}, \mathcal{K}_\subW \rangle.
\end{equation}

\begin{remark} \label{R:SemiBasic}
We say that a 2-form $B$ is {\it semi-basic with respect to the projection} $\tau:\M\to Q$ if ${\bf i}_X B = 0$ for all $X\in \mathfrak{X}(\M)$ such that $T\tau(X) = 0$. It was observed in \cite{PL2011,Naranjo2008} that if $B$ is semi-basic, then  $(\Omega_\subM - B)|_\C$ is still nondegenerate.  In our case,   since $\langle {\mathcal J}, \mathcal{K}_\subW \rangle$ is semi-basic, the 2-form $\Omega_{\mbox{\tiny{$\mathcal{J}\!\mathcal{K}$}}}$ is nondegenerate on $\C$ (see \cite[Sec.5.1]{Jac}).

\end{remark}

\begin{proposition}[\cite{Jac}] \label{Prop:OmegaJK-Moment}
 For $\eta\in \mathfrak{g}$, let us denote by $\xi= P_{\mathfrak{g}_\subSS}(\eta)\in \Gamma(\g_\subSS)$ where $P_{\mathfrak{g}_\subSS}: \mathfrak{g} \times Q \to \mathfrak{g}_\subSS$ is the projection associated to decomposition \eqref{Eq:gS+gW}.  If the vertical complement $W$ satisfies the vertical-symmetry condition then 
 $$
 {\bf i}_{\xi_\subM} \Omega_{\mbox{\tiny{$\mathcal{J}\!\mathcal{K}$}}} |_\C = d \langle J^{\emph\nh}, \xi\rangle|_\C.
 $$
\end{proposition}

In other words, Proposition \ref{Prop:OmegaJK-Moment} says that each $\eta\in \g$ induces a horizontal gauge momentum $j = \langle J^\nh, P_{\mathfrak{g}_\subSS}(\eta) \rangle$  of any vector field $Y_f$ given by ${\bf i}_{Y_f} \Omega_{\mbox{\tiny{$\mathcal{J}\!\mathcal{K}$}}} |_\C = d f|_\C$ for $f\in C^\infty(\M)^G$.

On the other hand, it is important to note that the function $j= \langle J^\nh, P_{\mathfrak{g}_\subSS}(\eta) \rangle$ is not necessarily a first integral of the nonholonomic dynamics.

\begin{proposition}\label{P:Obstruction} If $W$ satisfies the vertical-symmetry condition, then for $\xi= P_{\mathfrak{g}_\subSS}(\eta)\in \Gamma(\g_\subSS)$, $\eta\in \mathfrak{g}$ and $j=\langle J^{\emph\nh}, \xi \rangle \in C^\infty(\M)$ we have that
$$
X_{\emph\nh}(j)  =  \langle \mathcal{J}, \mathcal{K}_\subW \rangle (\xi_\subM, X_{\emph\nh}).
$$
\end{proposition}

\begin{proof} From Proposition \ref{Prop:OmegaJK-Moment},
 \begin{equation*}
  \begin{split}
 X_\nh( \langle J^{\nh}, \xi\rangle) = & \ \Omega_{\mbox{\tiny{${\mathcal J}\!{\mathcal K}$}}}(\xi_\subM, X_\nh) = \Omega_\subM (\xi_\subM, X_\nh) + \langle{\mathcal J}, \mathcal{K}_\subW \rangle (\xi_\subM, X_\nh) \\
   = & \ -d\Ham_\subM (\xi_\subM)  + \langle{\mathcal J}, \mathcal{K}_\subW \rangle (\xi_\subM, X_\nh) =  \langle{\mathcal J}, \mathcal{K}_\subW \rangle (\xi_\subM, X_\nh),
  \end{split}
 \end{equation*}
where we used the $G$-invariance of the hamiltonian $\Ham_\subM$. 

\end{proof}

Let us now consider $\{\eta_{k+1}, ... , \eta_N\}$ a basis of the Lie algebra $\mathfrak{w}$ and complete the basis so that $\{\eta_1, ... , \eta_k, \eta_{k+1} ,..., \eta_N\}$ is a basis of $\mathfrak{g}$. We denote by 
\begin{equation}\label{Eq:basis-xi-i}
\xi_i := P_{\mathfrak{g}_\subSS} (\eta_i) \in \Gamma(\g_\subSS)
\end{equation}
for $i=1,...,k$ (observe that $P_{\mathfrak{g}_\subSS} (\eta_j) = 0$ for $j \geq k+1$). Then $\{\xi_1,...,\xi_k\}$ is a (global) basis of sections of $\g_\subSS$.
Let us denote $j_i\in C^\infty(\M)$, $i=1,...,k$, the functions given by 
\begin{equation}\label{Def:j}
j_i :=  \langle J^{\nh}, \xi_i \rangle = \langle J^{\nh}, P_{\g_\S}(\eta_i) \rangle.
\end{equation}
Observe that each $j_i\in C^\infty(\M)$ is linear on the fibers (linear in momenta variables).  Then, for a nonholonomic system $(\M, \pi_\nh,\Ham_\subM)$ with a $G$-symmetry and a complement of the constraints satisfying the vertical-symmetry condition, we obtain

\begin{proposition} \label{P:gaugeMomentum}
The function $J$ on $\M$ is a horizontal gauge momentum of $X_{\emph\nh}$ if and only if  $J = \Sigma_{i=1}f_i j_i$ for $j_i \in C^\infty(\M)$ defined in \eqref{Def:j} and  $f_i\in C^\infty(Q)$ satisfying the following first-order linear partial differential equation 
\begin{equation}\label{Eq:DifEq}
 \Sigma_{i=1} \left(\, j_i df_i (X_{\emph\nh}) + f_i \langle {\mathcal J}, \mathcal{K}_\subW\rangle ((\xi_i)_\subM, X_{\emph\nh}) \, \right)= 0,
 \end{equation}
 where $\xi_i\in \Gamma(\g_\subSS)$ are  given in \eqref{Eq:basis-xi-i}. The associated horizontal gauge symmetry is $\chi=\Sigma_{i=1} f_i \xi_i$.
\end{proposition}

\begin{proof}
 If $J$ is a horizontal gauge momentum, then there is a section $\chi\in \Gamma(\g_\subSS)$ such that $J={\bf i}_{\chi_\subM} \Theta_\subM$.   Since $\{\xi_1,...,\xi_k\}$ (see \eqref{Eq:basis-xi-i}) is a basis of $\Gamma(\g_\subSS)$, then $\chi = \Sigma_{i=1} f_i \xi_i$ for $f_i \in C^\infty(Q)$ and thus $J= \Sigma_{i=1} f_i {\bf i}_{(\xi_i)_\subM}\Theta = \Sigma_{i=1} f_i j_i$.  Moreover, since $X_\nh(J)=0$ then $\Sigma_{i=1}\left(\, f_i dj_i(X_\nh)+ j_idf_i(X_\nh)\, \right)=0$. Finally, from Prop.~\ref{P:Obstruction}, we obtain that $(f_1,...,f_k)$ satisfy \eqref{Eq:DifEq}.  The converse is straightforward since if $J = \Sigma_{i=1} f_i j_i$, then $J={\bf i}_{\chi_\subM}\Theta_\subM$ for  $\chi=\Sigma_{i=1} f_i \xi_i$.  The differential equation \eqref{Eq:DifEq} guarantees that $X_\nh(J)=0$. 
\end{proof}

As opposed to what happens in hamiltonian systems, for nonholonomic systems the nonholonomic vector field associated to the conserved quantity $J$, i.e., $\pi_\nh^\sharp(dJ)$, is not necessarily a vertical vector field with respect to the orbit projection $\rho:\M\to\M/G$.  

\begin{proposition}\label{P:Lambda}  Let $j_i= \langle J^{\emph\nh}, \xi_i \rangle \in C^\infty(\M)$ be the components of the nonholonomic momentum map defined in \eqref{Def:j}.   If $J= \Sigma_{i=1} f_i\, j_i \in C^\infty(\M)$ and  $\chi = \Sigma_{i=1} f_i\, \xi_i\in\Gamma(\g_\subSS)$ for $f_i\in C^\infty(Q)$,  then
 $$\pi_{\emph\nh}^\sharp(\Lambda - dJ)= \chi_\subM $$ where $\Lambda$ is the 1-form on $\M$ given by 
\begin{equation} \label{Eq:Lambda}
 \Lambda =\Sigma_{i=1} \left(\,  j_i df_i + f_i {\bf i}_{(\xi_i)_\subM} \langle \mathcal{J}, \mathcal{K}_\subW \rangle \, \right).
 \end{equation}
\end{proposition}

\begin{proof} From the definition of $\Omega_{\mbox{\tiny{${\mathcal J}\!{\mathcal K}$}}}$ (see  \eqref{Eq:OmegaJK}) and Proposition \ref{Prop:OmegaJK-Moment}, we have that 
\begin{equation*}
 \begin{split}
{\bf i}_{\chi_\subM} \Omega_\subM |_\C = & \, \Sigma_{i=1} f_i \left( {\bf i}_{(\xi_i)_\subM} \Omega_{\mbox{\tiny{${\mathcal J}\!{\mathcal K}$}}} -{\bf i}_{(\xi_i)_\subM} \langle \mathcal{J}, \mathcal{K}_\subW \rangle \right)|_\C = \Sigma_{i=1} f_i \left( dj_i -{\bf i}_{(\xi_i)_\subM} \langle \mathcal{J}, \mathcal{K}_\subW \rangle \right)|_\C \\ 
= & \, (dJ- \Lambda)|_\C.  
 \end{split}
\end{equation*}
Therefore, from \eqref{Eq:HamVectorField}, we obtain that  $\pi_{\nh}^\sharp(\Lambda - dJ)= \chi_\subM $.
\end{proof}

If $J\in C^\infty(\M)$ is $G$-invariant then we denote by $\bar J \in C^\infty(\bar\M)$ the function such that $\rho^*\bar J=J$.  If, moreover,  $\Lambda =0$ then $\bar J\in C^\infty(\bar\M)$ is a {\it Casimir} of the reduced bracket $\{ \cdot , \cdot \}_\red$, that is, for all $f\in C^\infty(\bar\M)$, 
$$
\{f, \bar J\}_\red \circ \rho =  \{ \rho^*f, J\}_\nh =\chi_\subM (\rho^*f) = 0.
$$ 
As we will see later, the amount of Casimir functions of the reduced bracket will be fundamental to conclude that the bracket is Poisson (or twisted Poisson).

Next, we show with a simple example (the nonholonomic particle) how to proceed to find the horizontal gauge momentum using Prop. \ref{P:gaugeMomentum}. Moreover, we will see that in this case, not only the horizontal gauge momentum $J$ is $G$-invariant but also  $\Lambda =0$. Later in Section \ref{S:Body-of-revolution} we study the case of bodies of revolution rolling on a plane, including the Routh sphere and the axisymmetrical ellipsoid.

\begin{example}[{\bf The nonholonomic particle}] \label{Ex:NHParticle}  Consider a particle in $Q=\R^3$ with coordinates $(x,y,z)$
 restricted to the nonholonomic constraints $\dot z = y \dot x$ and where the 
 lagrangian is the canonical kinetic energy metric on $\R^3$. The distribution $D$ has fibers $D_q = \textup{span}\{ X_x := \partial_x+y\partial z, \, \partial_y  \}$ and  we consider the (free and proper) $\R^2$-action so that $V= \textup{span}\{ \partial_x,\partial_z \}$.  If we choose $W=\textup{span}\{ \partial_z\}$ then  $\{ X_x , \partial_y, \partial_z\}$ is a basis of $TQ$ adapted to $D\oplus W$. On $T^*Q$ we have the dual basis $\{dx,dy,\epsilon = dz-ydx\}$ and we  denote by $(p_x,p_y,p_y)$ its associated coordinates on $T_q^*Q$. Then
 the constraint submanifold is $\M = \{(x,y,z; p_x ,p_y , p_z) \ : \  p_z = \frac{y}{1+y^2}p_x \}, $
 and the restricted hamiltonian is $\Ham_\subM = \frac{1}{2} ( \frac{p_x^2}{1+y^2}+p_y^2 )$.   The nonholonomic vector field is 
 \[
 X_\nh = \frac{p_x}{1+y^2} X_x + p_y\frac{\partial}{\partial y} + \frac{yp_xp_y}{1+y^2} \frac{\partial}{\partial p_x}.
 \] 
The Lie algebra $\g = \R^2$ is split into $\mathfrak{g}_\subSS|_{(x,y,z)} = \textup{span}\{(1,y)\}$ (note that $(1,y)_\subQ = X_x$) and $\mathfrak{w} = \textup{span}\{(0,1)\}$ and hence $W$ satisfies the vertical-symmetry condition. 
Now, using Prop.~\ref{P:gaugeMomentum}, we find the horizontal gauge momentum of $X_\nh$.  The coordinate expression of $\langle \mathcal{J},\mathcal{K}_\subW\rangle$ is $\langle \mathcal{J},\mathcal{K}_\subW\rangle = \iota^*(p_z) d\epsilon$, a straightforward computation shows us that 
$$
\langle \mathcal{J},\mathcal{K}_\subW\rangle =  \frac{y}{1+y^2}p_x dx\wedge dy.
$$ 
Then, observe that the constant section $(1,0)\in \g\times Q$ induces the section  $P_{\mathfrak{g}_\subSS} ((1,0) ) =  (1,y)$ in $\g_\subSS$ and  
the function $j = \langle J^\nh, (1,y)\rangle = p_x$.
Therefore the horizontal gauge momentum is $J= f\,\langle J^\nh, (1,y)\rangle = f p_x$, where $f\in \C^\infty(\R^3)$ satisfies that 
\begin{equation}\label{Eq:NH-MomentEq} 
p_x df(X_\nh) + f \frac{y}{1+y^2} p_x p_y = 0, 
\end{equation}
since $\langle \mathcal{J},\mathcal{K}_\subW\rangle (( P_{\mathfrak{g}_\subSS}((1,0)))_\subM, X_{\nh})= \frac{y}{1+y^2} p_x p_y.$ 
As a first attempt, let us solve \eqref{Eq:NH-MomentEq} for $f$ a $G$-invariant function. Then $f=f(y)$ satisfies that $df(X_\nh) = f' p_y$, hence \eqref{Eq:NH-MomentEq} is written as $f' + f\frac{y}{1+y^2} = 0$ and thus $f(y) = \frac{1}{\sqrt{1+y^2}}$. 

Now, using that $j=p_x$, $df = - \tfrac{y}{1+ y^2} f dy$ and ${\bf i}_{X_x} \langle \mathcal{J}, \mathcal{K}_\subW \rangle = \tfrac{y}{1+ y^2} p_x dy$ we check that $\Lambda = 0$ (see \eqref{Eq:Lambda}). By Prop.~\ref{P:Lambda}, $\pi_\nh^\sharp(dJ) = f(y) X_x \in \Gamma(\S)$. 
Since $J= f(y) p_x$ is $G$-invariant, we conclude that $J$ is a Casimir of the reduced bracket $\{\cdot, \cdot\}_\red$ on $\M/G$.

{\hfill$\diamond$}

\end{example}

However, $\Lambda \neq 0$ for most of the examples such as the Chaplygin ball, bodies of revolution rolling on a plane and even in the case of the homogeneous ball rolling on a convex surface (see Example \ref{Ex:ChapBall}, Section \ref{S:Body-of-revolution} below and \cite{Yapu} respectively) .  

\subsection{Gauge transformations and their consequences}\label{Ss:Gauge}

In this section, we will introduced {\it gauge transformations by 2-forms $B$} \cite{SW} of $\pi_\nh$ in order to produce new brackets $\pi_\B$ describing the dynamics (i.e., $\pi_\B^\sharp(d\Ham_\subM)=-X_\nh$) but with different properties than the nonholonomic bracket. In fact, we will see a way to produce a bracket in such a way that the hamiltonian vector field of the first integral $J$ (with respect to this new bracket) becomes vertical.

\begin{definition}[\cite{PL2011}] \label{Def:Gauge}
 A 2-form $B$ on $\M$ defines a {\it dynamical gauge transformation} of the nonholonomic system $(\M, \pi_\nh, \Ham_\subM)$ if
 \begin{enumerate}\item[$(i)$] $(\Omega_\subM + B)|_\C$ is nondegenerate,
  \item[$(ii)$]  ${\bf i}_{X_\nh} B =0$.
 \end{enumerate}
\end{definition}

If a 2-form $B$ satisfies condition $(i)$ of the above definition, then there is a new bivector field $\pi_\B$ on $\M$ defined by the relation 
\begin{equation}\label{Eq:B-NHbracket}
\pi_\B^\sharp(\alpha) = - X \quad \mbox{if and only if}  \quad {\bf i}_{X} (\Omega_\subM + B)|_\C = \alpha |_\C,
\end{equation}
for $\alpha$ a 1-form on $\M$. 
In this case, we say that $\pi_\B$ is the induced bracket by the gauge transformation by $B$ of $\pi_\nh$. We also say that $\pi_\B$ and $\pi_\nh$  are {\it gauge related}. 

Condition $(ii)$ of Def.~\ref{Def:Gauge}  guarantees that this new bivector field $\pi_\B$ describes the dynamics, i.e., the nonholonomic vector field $X_\nh$ is also given by 
\begin{equation}\label{Eq:B-dynamics}
X_\nh = - \pi_\B^\sharp(d\Ham_\subM),
\end{equation}
and in this case, $\pi_\B$ and $\pi_\nh$  are {\it dynamically gauge related}.

The first time that the idea of modifying the nonholonomic bracket by a 2-form appeared was in \cite{Naranjo2008}.

Using Remark \ref{R:SemiBasic}, we see that any semi-basic 2-form $B$ (with respect to the bundle $\tau: \M\to Q$) will satisfy automatically condition $(i)$ of Def.~\ref{Def:Gauge}.

In a similar way as it was done in \eqref{Eq:ReducedBracket}, in the presence of symmetries, we can define a reduced bracket $\{\cdot, \cdot \}_\red^\B$ on $C^\infty(\bar\M)$  induced by $\{\cdot, \cdot \}_\B$. That is, if the 2-form $B$ is $G$-invariant, then $\pi_\B$ is $G$-invariant and thus there is an almost Poisson structure $\{\cdot, \cdot \}_\red^\B$ on the differential space $\bar\M$, defined for $f,g\in C^\infty(\bar\M)$, by 
\begin{equation}\label{Eq:GaugeReducedBracket}
\{f,g\}_\red^\B \circ \rho = \{\rho^*f, \rho^*g\}_\B. 
\end{equation}

Since the gauge related bivector field $\pi_\B$ depends only on the values of the 2-form $B$ on $\C$, we can assume that $B$ satisfies that $B|_\W = 0$. Analogous to the formula \eqref{Eq:Jac-pi_red}, it was proven in \cite{Jac} that the failure of the Jacobi identity of $\{\cdot, \cdot\}_\red^\B$ on $\bar\M$ is encoded in the 3-form $d\langle {\mathcal J}, \mathcal{K}_\subW \rangle - dB$ defined on $\M$. That is, for $f,g,h \in C^\infty(\bar\M)$,
\begin{equation}\label{Eq:Jac-pi_Bred}
cyclic \left[ \, \{f,\{g,h\}_{\red}^{\B}\}_{\red}^{\B}  \circ \rho \, \right]   =  [d\langle {\mathcal J}, \mathcal{K}_\subW \rangle - dB ] ( \pi_\B^\sharp(d\rho^*f), \pi_\B^\sharp(d\rho^*g),\pi_\B^\sharp(d\rho^*h) ).
\end{equation}
Hence, $\{\cdot, \cdot \}_\red^\B$ is a Poisson bracket on $\bar\M$ if 
$$
[d\langle {\mathcal J}, \mathcal{K}_\subW \rangle - dB ]|_{U_\B} = 0
$$
where $U_\B$ is the distribution on $\M$ given by 
\begin{equation}\label{Eq:U_B}
 U_\B=\textup{span}\{\pi_\B^\sharp(dF), \mbox{ where } F\in C^\infty(\M)^G\}
\end{equation}

\begin{remark}
 The distribution $U_0=\textup{span}\{\pi_\nh^\sharp(dF), \mbox{ where } F\in C^\infty(\M)^G\}$ already appeared in \cite{BS93} in order to study the reduced dynamics (for free and proper actions).
\end{remark}

Now, going back to the ideas of Section \ref{Ss:Symm-Conservation}, we observe from Proposition~\ref{P:Lambda} that given a horizontal gauge momentum $J$, the vector field $\pi_\nh^\sharp(dJ)$ is not necessarily vertical (with respect to the orbit projection $\rho:\M \to \bar\M$).
The idea of applying a dynamical gauge transformation on $\pi_\nh$ is to obtain a bivector field $\pi_\B$ so that $\pi^\sharp_\B(dJ)$ is vertical.
 In this case, if the conserved quantity $J$ is also $G$-invariant, then $\bar J\in C^\infty(\bar\M)$ (defined by $\rho^*\bar J = J$) becomes a Casimir of the reduced bracket $\{ \cdot, \cdot \}_\red^\B$ on $\bar\M$.

\begin{theorem}\label{T:Casimirs} 
Consider $(\M, \pi_{\emph\nh},\Ham_\subM)$ a nonholonomic system with a $G$-symmetry (where $G$ acts properly on $Q$) and let $J$ be a horizontal gauge momentum of $X_{\emph\nh}$ with $\chi\in\Gamma(\g_\subSS)$ the corresponding horizontal gauge symmetry.   If $B$ defines a dynamical gauge transformation such that ${\bf i}_{\chi_\subM}B|_\C= \Lambda|_\C$, then $\pi_\B$ satisfies \eqref{Eq:B-dynamics} and verifies that 
 $$
 \pi_\B^\sharp(dJ)= - \chi_\subM\in\Gamma(\V).
 $$
 Moreover if $B$ and $J$ are $G$-invariant then $\bar J\in C^\infty(\bar\M)$ is a Casimir of the reduced bracket $\{\cdot, \cdot\}_{\emph\red}^\B$.
\end{theorem}

\begin{proof}
From Prop.~\ref{P:Lambda} we have that ${\bf i}_{\chi_\subM} \Omega_\subM|_\C = (dJ-\Lambda)|_\C$. So, if ${\bf i}_{\chi_\subM} B|_\C = \Lambda|_\C$ then ${\bf i}_{\chi_\subM} (\Omega_\subM +B)|_\C = dJ|_\C$.  Using \eqref{Eq:B-NHbracket} we obtain that $\pi_\B^\sharp(dJ)= - \chi_\subM$. Now, if $B$ is $G$-invariant, then $\{\cdot, \cdot\}_\B$ is $G$-invariant and thus for any $f\in C^\infty(\bar\M)$ we have
$$
\{f, \bar J\}^\B_\red \circ \rho =  \{\rho^*f, J\}_\B =\chi_\subM (\rho^*f) = 0.
$$  
Hence, $\bar{J}\in C^\infty(\bar\M)$ is a Casimir of $\{\cdot, \cdot\}_\red^\B$. 
\end{proof}

\begin{remark}(Involution of first integrals) \label{R:Involution}
 Suppose that the nonholonomic system has two $G$-invariant horizontal gauge momenta $J_1$ and $J_2$; in general, $\{J_1, J_2\}_\nh \neq 0$. However, if we choose a dynamically gauge related bracket $\{\cdot, \cdot\}_\B$ for which ${\bf i}_{(\chi_1)_\subM} B= \Lambda_1$ and ${\bf i}_{(\chi_2)_\subM} B= \Lambda_2$, then this new bracket puts the functions $J_1$ and $J_2$ in involution, i.e., $\{J_1, J_2\}_\B=0$. 
\end{remark}

\begin{remark} If the nonholonomic system $(\M, \pi_\nh,\Ham_\subM)$ with a $G$-symmetry has an horizontal gauge momentum $J$, then by the definition, $J$ is a first integral of the nonholonomic vector field $X_\nh$. However, it may happen that $J$ is not a first integral of all the vector fields in $U_0$ (as it is the case of a hamiltonian system). The choice of a 2-form $B$ such that ${\bf i}_{\chi_\subM}B|_\C=\Lambda|_\C$ produces a new bivector $\pi_\B$ for which $J$ is a first integral for all vector fields in $U_\B$.
\end{remark}

\subsection{The case of a proper and free action} \label{Ss:Regular}

In this section we study the case of a nonholonomic system on the manifold $Q$ described by a lagrangian $L$ and a nonintegrable distribution $D$ with a $G$-symmetry given by the free and proper action of a Lie group $G$. As a consequence, we will see that if $\S$ is generated by vector fields of the type $\pi_\B^\sharp(dJ_i)$ for $J_i$ $G$-invariant functions on $\M$, then we can conclude that the reduced bracket has almost symplectic leaves, i.e., it is {\it twisted Poisson}  \cite{SW}. 

Suppose that the action of the Lie group $G$ on $Q$ is free and proper. Then
\begin{enumerate}
 \item[$(i)$] the quotient space $\M/G$ is a manifold.
 
 \item[$(ii)$]  The vertical distributions $\V_m\subset T_m\M$ and $V_{\tau(m)}\subset T_{\tau(m)}Q$ are of constant rank and isomorphic under the projection $T\tau : T\M \to TQ$ (respectively, the distributions  $\S_m = \C_m \cap \V_m\subset T_m\M$ and $S_{\tau(m)} = D_{\tau(m)} \cap V_{\tau(m)} \subset T_{\tau(m)}Q$ are isomorphic and constant rank). The rank of $S$ (or $\S$) is exactly the rank of $\g_\subSS$ (defined in Def.~\ref{Def:gS}). 
 
 \item[$(iii)$] The failure of the Jacobi identity \eqref{Eq:Jac-pi_red} for $\{\cdot, \cdot\}_\red^\B$ on $\M/G$ can be written in terms of the reduced bivector field $\pi_\red^\B$,  
\begin{equation}\label{Eq:Regular-Jac}
\tfrac{1}{2}[\pi_\red^\B, \pi_\red^\B] = - \rho_*[\pi_\B^\sharp(d( \langle \mathcal{J}, \mathcal{K}_\subW \rangle - B))], 
\end{equation}
where, for $\alpha,\beta,\gamma\in \Omega^1(\M/G)$,  $\rho_*[\pi_\B^\sharp(d( \langle \mathcal{J}, \mathcal{K}_\subW \rangle - B))](\a,\beta,\gamma)(\rho(m)) := \pi_\B^\sharp(d( \langle \mathcal{J}, \mathcal{K}_\subW \rangle - B))(\rho^*\a,\rho^*\beta,\rho^*\gamma) (m)$.
We conclude also that if $d\langle {\mathcal J}, \mathcal{K}_\subW \rangle$ (respectively $\langle {\mathcal J}, \mathcal{K}_\subW \rangle$) is a well defined 3-form (resp. 2-form) on $\bar\M$ then $\pi^\B_\red$ is a twisted Poisson bracket on $\bar\M$ and formula \eqref{Eq:Jac-pi_red} reduces to 
 $
 \tfrac{1}{2}[\pi_\red , \pi_\red]  = - \pi_\red^\sharp ( \Phi)
 $
 where $\Phi$ is the 3-form on $\bar\M$ such that $\rho^*\Phi = d(\langle {\mathcal J}, \mathcal{K}_\subW \rangle-B)$.  
 
 \end{enumerate}

%In what follows we will show that if $\S$ is generated by ``hamiltonian'' vector fields of the type $\pi_\B^\sharp(dJ_i)$ for $J_i$ $G$-invariant functions on $\M$, then he reduced bivector field $\pi_\red^\B$ on $\M/G$ has an integrable characteristic distribution. % the system has $k=\textup{rank}\, \S$ functionally independent functions $J_i$ such that  $\pi_\B^\sharp(dJ_i)\in\Gamma(\V)$ then the reduced bivector field $\pi_\red^\B$ on $\M/G$ has an integrable characteristic distribution. 

\begin{proposition} \label{P:Involutivity}
Let $\pi_\B$ be the gauge related bivector to $\pi_{\emph\nh}$ by a $G$-invariant 2-form $B$ and suppose that $G$ acts freely and properly on $Q$.  If there are $G$-invariant functions  $J_1,..., J_k$  on $\M$ such that $\{\pi^\sharp_\B(dJ_i)\}_{i,...,k}$ generate $\S$, then the characteristic distribution of the reduced bivector field $\pi_{\emph\red}^\B$ is involutive. 
\end{proposition}

\begin{proof}
The characteristic distribution of $\pi_\red^\B$ on $\M/G$ is generated by the vector fields $(\pi_\red^\B)^\sharp(df)$ for $f\in C^\infty(\M/G)$.  For $f,g\in C^\infty(\M/G)$, we have the formula 
 \begin{equation}\label{Eq:Proof-Invol}
 [(\pi_\red^\B)^\sharp(df), (\pi_\red^\B)^\sharp(dg)] = (\pi_\red^\B)^\sharp(d\{f,g\}_\red^\B) + \tfrac{1}{2}{\bf i}_{df\wedge dg} [\pi_\red^\B, \pi_\red^\B].
 \end{equation}
 Using \eqref{Eq:Regular-Jac}, we see that for all $h\in C^\infty(\M/G)$,
 $$
 \tfrac{1}{2}{\bf i}_{df\wedge dg} [\pi_\red^\B, \pi_\red^\B] (dh) \circ \rho = - \pi_\B^\sharp \left( {\bf i}_{\pi_\B^\sharp (d\rho^*f)\wedge \pi_\B^\sharp (d\rho^*g)} d( \langle \mathcal{J}, \mathcal{K}_\subW \rangle - B)\right) (d\rho^*h).
 $$
 Define $\Theta$ the 1-form on $\M$ such that 
 $$
 \Theta |_\C = {\bf i}_{\pi_\B^\sharp (d\rho^*f)\wedge \pi_\B^\sharp (d\rho^*g)} d( \langle \mathcal{J}, \mathcal{K}_\subW \rangle - B) |_\C \quad \mbox{and} \quad \Theta|_\W = 0.
 $$  %Observe that $\pi_\B^\sharp(\Sigma) = \pi_\B^\sharp( {\bf i}_{\pi_\B^\sharp (d\rho^*f)\wedge \pi_\B^\sharp (d\rho^*g)} d( \langle \mathcal{J}, \mathcal{K}_\subW \rangle +B) )$.  
 Next we see that $\Theta |_\S = 0$.  In fact, since $\{\pi^\sharp_\B(dJ_i)\}_{1,...,k}$ generate $\S$, using \eqref{Eq:Regular-Jac}, we see that
 $
 \Theta(\pi^\sharp_\B(dJ_i)) = cyclic \left[ \, \{f,\{g,\bar{J}_i\}_\red^\B\}_\red^\B \circ \rho \, \right] = 0,
 $ for $\bar{J}_i\in C^\infty(\M/G)$ such that $\rho^*\bar{J}_i = J_i$. 
 Since also $\Theta|_\W = 0$ we obtain that $\Theta|_\V = 0$. But we know that $\Theta$ is $G$-invariant and therefore we conclude that $\Theta$ is a basic 1-form; we denote by $\bar \Theta$ the 1-form on $\M/G$ such that $\rho^*\bar\Theta = \Theta$.  Then \eqref{Eq:Proof-Invol} reads
 $$
 [(\pi_\red^\B)^\sharp(df), (\pi_\red^\B)^\sharp(dg)] = (\pi_\red^\B)^\sharp(d\{f,g\}_\red^\B) - (\pi_\red^\B)^\sharp (\bar\Theta),
 $$
 and we conclude that  $ [(\pi_\red^\B)^\sharp(df), (\pi_\red^\B)^\sharp(dg)]$ is a section of the characteristic distribution of $\pi_\red^\B$ and so it is involutive.  
 
\end{proof}

In other words, Proposition \ref{P:Involutivity} is proving that the distribution  $T\!\rho (U_\B)$ on $\M/G$ is involutive, where $U_\B$ is the distribution defined in \eqref{Eq:U_B}. 

%\begin{remark}
 %One can weaken the hypothesis of Prop.~\ref{P:Involutivity} by assuming that each point $m\in\M$ admits a neighborhood $\mathcal{U}$ for which there exist independent $G$-invariant functions  $J_1,..., J_k$  on $\mathcal{U}$ such that $\{\pi^\sharp_\B(dJ_i)\}_{i,...,k}$ is a basis of sections of $\S|_{\mathcal U}$. In this case, we still obtain the involutivity of the characteristic distribution of $\pi_{\red}^\B$. 
%\end{remark}

It was proven in \cite[Corollary~3]{PL2011} that any regular bivector field on a manifold with an integrable characteristic distribution, is twisted Poisson. Therefore,  

\begin{corollary} \label{C:regular-twisted}
Under the hypothesis of Prop.~\ref{P:Involutivity} and if moreover the reduced bivector $\pi_{\emph\red}^\B$ on $\M/G$ is regular, then $\pi_{\emph\red}^\B$ is twisted Poisson.
\end{corollary}

\begin{proof}
Since  Prop.~\ref{P:Involutivity} asserts that the characteristic distribution of $\pi_\red^\B$ defined on the manifold $\M/G$ is involutive, then this Proposition is just a consequence of \cite[Corollary~3]{PL2011}.
\end{proof}

In our context, the candidates to be the functions $\{J_i\}_{i=1,..,k}$ such that $\{\pi_\B^\sharp(dJ_i)\}_{i=1,..,k}$ generate $\S$, are the horizontal gauge momenta.  

More precisely, as we did in Sec.~\ref{Ss:Symm-Conservation}, if $k=\textup{rank}\,\g_\subSS$ let us consider $\{\eta_1, ... , \eta_k, \eta_{k+1} ,..., \eta_N\}$ a basis of the Lie algebra $\mathfrak{g}$ so that $\{\eta_{k+1}, ... , \eta_N\}$  is a basis of $\mathfrak{w}$.
Then by \cite[Prop.~4]{BalSan} we have that the components of the nonholonomic momentum map 
$$
\{j_1, ..., j_k\} \quad \mbox{with } j_i=\langle J^\nh, P_{\g_\subSS}(\eta_i)\rangle \in C^\infty(\M), 
$$  
(see \eqref{Def:j}),  are (functionally) independent functions. 
Moreover, by Prop.~\ref{P:gaugeMomentum} we observe that if the differential equation \eqref{Eq:DifEq} has $k$ (functionally independent) solutions, then we have $k$ independent horizontal gauge momenta $J_1,..., J_k$. 

Therefore, if the system has $k=\textup{rank}\, \S$ $G$-invariant horizontal gauge momenta $\{J_1,...,J_k\}$ such that $\pi_\B^\sharp(dJ_i)\in\Gamma(\V)$ then the reduced bivector field $\pi_\red^\B$ on $\M/G$ has an integrable characteristic distribution.

\begin{example}[{\bf Nonholonomic particle}] \label{Ex:NH-Part2}  Following Example~\ref{Ex:NHParticle}, we see that the $\R^2$-action is free and proper and moreover,  the reduced bivector field $\pi_\red$ is regular on the manifold $\M/G\simeq \R^3$. We computed the horizontal gauge momentum $J= f(y)p_x$ and we saw that $\pi_\nh^\sharp(dJ) = - f(y)X_x \in \Gamma(\S)$. Now, since $S$ has constant rank 1, by Corollary~\ref{C:regular-twisted}, $\pi_\red$ is Poisson (in fact, the Corollary says that  $\pi_\red$ is twisted Poisson but since $\textup{dim}\M/G = 3$ then it is Poisson).  

{\hfill$\diamond$}
\end{example}

\begin{example}[{\bf Chaplygin ball}] \label{Ex:ChapBall}
The Chaplygin ball is the classical example of an inhomogeneous ball rolling without slipping on a plane \cite{BorisovMamaev,Duistermaat,Naranjo2008}. First we observe that the system has a symmetry induced by the free and proper action of the Lie group $G=SE(2)$ and that the dimension assumption is satisfied. Following \cite{Jac}, the $G$-invariant vertical complement of the constraints is related with the $\R^2$-translational symmetry of the system (observe that $G=SO(1)\times \R^2$). The complement $W$ satisfies the vertical-symmetry condition and $\mathfrak{w}$ is identified with $\R^2$.

Since the Lie algebra $\g$ is isomorphic to $\R^3$ we consider any basis $\{\eta, \eta_1, \eta_2\}$ of $\g$, where $\{\eta_1, \eta_2\}$ is a basis of the subalgebra $\mathfrak{w} = \R^2$.  
It is computed in \cite{Jac} that the section  $\xi = P_{\g_\subSS}(\eta)$  satisfies that $\langle \mathcal{J}, \mathcal{K}_\subW \rangle(\xi_\subM, X_\nh)= 0$ and thus from Prop.~\ref{P:Obstruction} we have that $j= \langle J^\nh, \xi\rangle$ is a horizontal gauge momentum (alternatively we can see that $f=cte$ is a solution of \eqref{Eq:DifEq}). 

 In order to find a 2-form $B$ such that $\pi_\B^\sharp(dj) = - \xi_\subM$, we follow Theorem \ref{T:Casimirs} and we see that the $B$ has to satisfy  
 \begin{equation}\label{Ex:Chap:Lambda} 
{\bf i}_{\xi_\subM} B = \Lambda.
 \end{equation}
 In this case the 1-form $\Lambda$ (see Prop.~\ref{P:Lambda}) is given by $\Lambda = {\bf i}_{\xi_\subM} \langle \mathcal{J}, \mathcal{K}_\subW \rangle$, hence \eqref{Ex:Chap:Lambda}  is equivalent to asking that $\langle \mathcal{J}, \mathcal{K}_\subW \rangle - B$ has to be semi-basic with respect to the orbit projection $\M\to \M/G$ (since it is $G$-invariant by construction, we may ask that it is basic).  We observe that the 2-form $B$ found in \cite{Naranjo2008} is a 2-form for which  $\langle \mathcal{J}, \mathcal{K}_\subW \rangle - B$ is basic, see \cite{Jac}.   Observe also that $\textup{rank}\g_\subSS =1$.

Therefore, for that 2-form $B$ we have that $\xi_\subM = - \pi^\sharp_\B(dj)$ and thus $\{\pi^\sharp_\B(dj)\}$ is a (global) generator of $\Gamma(\S)$. Since $j\in C^\infty(\M)$ is $G$-invariant, then again by Thm.~\ref{T:Casimirs}, $j\in C^\infty (\M/G)$ is a Casimir function of the reduced bivector field  $\pi_\red^\B$ on $\M/G$.  Noting that the reduced bivector field $\pi_\red^\B$ is regular (see \cite{PL2011}) then following Corollary~\ref{C:regular-twisted}, it is a twisted Poisson bracket recovering the result in \cite{PL2011,Jac}).

{\hfill$\diamond$}

\end{example}

For a non necessarily free action, we will see in the following examples that the property that $\pi^\sharp_\B(dJ)= - \chi_\subM$ is also what will tell us that the reduced bracket is Poisson.

\section{Body of revolution} \label{S:Body-of-revolution}

In this section we discuss the motion of a strongly convex body of revolution which rolls without slipping on a horizontal plane under the influence of a constant vertical gravitational force of strength g.  We follow the notation and ideas from \cite{Cushman1998,Book:CDS}.

\subsection{Preliminaries}

Consider a strongly convex body $\mathfrak{B}$ with mass $m$ and with (non zero) principal moments of inertia $\mathbb{I}_1, \mathbb{I}_2$ and $\mathbb{I}_3$.  
Following \cite{Book:CDS}, the body $\mathfrak{B}$ is a body of revolution if it is geometrically and dynamically symmetric under rotations about a given axis, which in our case is chosen to be $e_3$. That is, the surface $\mathbf{S}$ of $\mathfrak{B}$ is invariant under rotations around $e_3$ and $\mathbb{I}_1 = \mathbb{I}_2$. 

The position of the center of mass of the body is represented by the coordinates $\veca\in \R^3$ and the relative position of the body is given by $g\in SO(3)$. 
The lagrangian $L:T(SO(3) \times \R^3)\to \R$ is of mechanical type and given by
\begin{equation}\label{Eq:Lagrangian}
 L((g,\veca, \vecOm, \dot\veca)) = \frac{1}{2} \langle \mathbb{I}\vecOm, \vecOm\rangle + \frac{1}{2}m \langle \dot\veca, \dot\veca\rangle - m\mbox{g}\langle \veca, {\bf e}_3 \rangle, 
\end{equation}
for $\vecOm = (\Omega_1, \Omega_2, \Omega_3)$ the angular velocity of the body in body coordinates and where $\langle \cdot , \cdot \rangle$ denotes the standard inner product in $\R^3$.

Let $s$ be the vector from the center of mass to a point on the surface $\mathbf{S}$ and we denote by $n(s)$ the inward unit normal to $\mathbf{S}$ at $s$. The Gauss map $n:\mathbf{S} \to S^2$, $s \mapsto n(s)$, is a diffeomorphism since $\mathbf{S}$ is smooth, compact and strongly convex. 
We denote by $\vecgamma= (\gamma_1,\gamma_2, \gamma_3)\in S^2$ the third row of the matrix $g\in SO(3)$ and consider the inverse of the Gauss map $s:S^2 \to \mathbf{S}$, $\vecgamma \mapsto s(\vecgamma)$ given by
\begin{equation}\label{Eq:s}
s(\vecgamma) = (\varrho(\gamma_3)\gamma_1, \varrho(\gamma_3)\gamma_2, \zeta(\gamma_3)), 
 \end{equation}
where $\varrho:(-1,1) \to \R$ and $\zeta:(-1,1) \to \R$ are smooth functions whose definition depend on the shape of the body of revolution, see \cite[Sec.~6.7.1]{Book:CDS}.
During the treatment of the examples, we will use the notation   $\varrho=\varrho(\gamma_3)$, $\zeta = \zeta(\gamma_3)$ and 
$$
s =s(\vecgamma) = \varrho. \vecgamma - L {\bf e}_3, \qquad \mbox{for} \ L=L(\gamma_3) =  \varrho. \gamma_3 - \zeta.
$$

The configuration space $Q$ is the submanifold of $SO(3)\times \R^3$ defined by 
$$
Q=\{(g,\veca)\in SO(3)\times \R^3 \ : \ a_3 = - \langle \vecgamma, s \rangle \},
$$ 
which is diffeomorphic to $SO(3)\times \R^2$ with coordinates $(g,\veca) = (g,a_1, a_2)$. The constraint $a_3 = \langle \vecgamma, s\rangle$ is the holonomic constraint representing the fact that the body is on a plane.  Let us consider the (local) basis of $TQ$ given by 
$\{ X_1^L, X_2^L, X_3^L, \partial_{a_1}, \partial_{a_2} \}$, where $X_i^L$ are the left invariant vector fields on $SO(3)$ and we denote the corresponding coordinates on $TQ$ by $(\vecOm, \dot a_1, \dot a_2)$.

The nonholonomic constraints are 
\begin{equation}   \label{Eq:Constraints}
 \vecOm \times s + {\bf b} = 0, 
\end{equation}
where ${\bf b} = g^T \dot\veca$ (i.e., the constraint can be written as $\dot\veca = - g . ( \vecOm\times s)$). 
Hence, using that $a_3 = - \langle \vecgamma, s \rangle$  we obtain the following relation: $-\varrho' + L' \gamma_3 = 0$. 

The constraint distribution $D$ on $Q$ is given by $D= \textup{span}\{ X_1, X_2, X_3\}$, where 
\begin{equation}\label{Eq:D}
\begin{split} 
X_i & := X_i^L+(\vecalpha\times s)_i \partial_{a_1} + (\vecbeta \times s)_i \partial_{a_2} + (\vecgamma \times s)_i \partial_{a_3} \\
& = X_i^L-(s\times Y)_i,  
\end{split}
\end{equation}
for $Y =  (Y_1,Y_2, Y_3)$, $Y_i = \alpha_i \frac{\partial}{\partial a_1} + \beta_i \frac{\partial}{\partial a_2} + \gamma_i \frac{\partial}{\partial a_3}$ and $\vecalpha,\vecbeta, \vecgamma$ being the three rows of the matrix $g\in SO(3)$. We denote by $ {\bf X}= (X_1, X_2, X_3) $ and by $\vecL = (\lambda_1, \lambda_2, \lambda_3)$ the (Maurer-Cartan) 1-forms on $SO(3)$ dual to the left invariant vector fields $\{X^L_1,X^L_2, X^L_3\}$. 
The constraint 1-forms are 
\begin{equation}\label{Eq:Constraints-1forms}
\epsilon^1 = da_1 -\langle \vecalpha, s\times \vecL \rangle \quad \mbox{and}\quad \epsilon^2= da_2 - \langle \vecbeta, s \times \vecL\rangle. 
\end{equation}

\subsection{The Lie group symmetry and the vertical complement of the constraints}

Following \cite{Book:CDS}, consider the 2-dimensional euclidean Lie group $E(2)$ given by 
$$
E(2)= \left\{ (h_\varphi, (x,y)) \in SO(3)\times \R^2 \ : \ h_\varphi =  \left(\begin{array}{ccc} \bar{h}_\varphi & & 0 \\ & & 0 \\ 0 & 0 & 1 \end{array}  \right) \mbox{ with } \bar{h}_\varphi =  \left(\begin{array}{cc} \cos\varphi & -\sin\varphi \\ \sin\varphi & \cos\varphi \end{array}  \right)\right\}.                                                    
$$ 
The system is invariant by the left $E(2)$-action on $Q$ defined by $\left( (h_\varphi,(x,y)), (g,\veca) \right) \mapsto (h_\varphi g, \\ \bar{h}_\varphi\veca + (x,y))$ (that is, this action leaves the lagrangian and the constraints invariant). Moreover, by the symmetry of the body, the system is also invariant by the right $S^1$-action on $Q$ given by $\left( h_\theta, (g,\veca) \right) \mapsto (gh_\theta^{-1}, \bar{h}_\theta \veca )$. Since both actions commute, we consider the Lie group $G= S^1\times E(2)$ as the symmetry group of the nonholonomic system.  

The associated Lie algebra is $\mathfrak{g} \simeq \R \times \R \times \R^2$ and the infinitesimal generator relative to the $S^1$-action is $(1; 0, {\bf 0})_\subQ = - X_3^L - a_2 \frac{\partial}{\partial a_1} + a_1 \frac{\partial}{\partial a_2}$, while relative to the $E(2)$-action we have
\begin{equation}\label{Eq:General-InfGen}
(0;1,{\bf 0})_\subQ = \langle \vecgamma, {\bf X}^L \rangle  - a_2 \frac{\partial}{\partial a_1} + a_1\frac{\partial}{\partial a_2}, \quad (0;0,(1,0))_\subQ = \frac{\partial}{\partial a_1}, \quad (0;0,(0,1))_\subQ = \frac{\partial}{\partial a_2}. 
\end{equation}

\begin{remark}
For $q=(g,\veca)$, where $g$ has the third row equal to $(0,0,\pm 1)$, $(1; 0, {\bf 0})_\subQ (q)= \mp (0;1,{\bf 0})_\subQ(q)$, so we see that the $G$-action is not free. 
\end{remark}

Note that $S=D \cap V = \textup{span}\{X_3, \langle \vecgamma, {\bf X}\rangle\}$ does not have constant rank. However, the dimension assumption \eqref{Eq:DimAssumptionTQ} is satisfied: from \eqref{Eq:D} and \eqref{Eq:General-InfGen} we see that $TQ=D+V$.
Therefore, we can choose a (smooth, rank 2) vertical complement $W$ of the constraints:  for each $(g,\veca)\in Q$,
\begin{equation}\label{Eq:Solids-D+W}
D_{(g,\veca)} = \textup{span}\{ {\bf X}= X_1, X_2, X_3 \} \quad \mbox{and}\quad W_{(g,\veca)} =\textup{span} \left\{ \frac{\partial}{\partial a_1}, \frac{\partial}{\partial a_2} \right\}.
\end{equation}

The vertical complement $W$ satisfies the vertical-symmetry condition, since 
$$
\mathfrak{w} = \textup{span}\{ (0; 0 ,(1,0)), (0;0,(0,1)) \}
$$ 
is a $Ad$-invariant subalgebra of $\mathfrak{g}$ and, at each $q\in Q$,  $\Psi_q|_{\mathfrak{w}} : \mathfrak{w} \to W_q$ is an isomorphism (see Def.~\ref{Def:VertSym}). On the other hand, from \eqref{Def:gS} we see that 
\begin{equation}\label{Eq:gs}
\mathfrak{g}_\subSS = \textup{span}\{  \xi_1 := ((1 ; 0 , (h_1, h_2)), \,  \xi_2 := ((0; 1, (g_1, g_2))\} 
\end{equation}
where $h_1 = a_2 +  \varrho \beta_3$, $h_2 = - a_1 -  \varrho \alpha_3$ and  $g_1 =  a_2 + \langle s, \vecbeta\rangle $, $g_2 = - a_1 - \langle s, \vecalpha \rangle$. 
The infinitesimal generators associated to $\xi_1$ and $\xi_2$ are vector fields with values in $S$ given by
\begin{equation} \label{Eq:InfinitesGenerators-Q}
{(\xi_1)}_\subQ = -X_3  \qquad \mbox{and}  \qquad {(\xi_2)}_\subQ = \langle \vecgamma, {\bf X} \rangle,
\end{equation}

Finally we set an adapted basis to the splitting $TQ= D\oplus W$ given by $\{{\bf X},\partial_{a_1}, \partial_{a_2} \}$ and its dual defined by $\{ \vecL, \epsilon^1, \epsilon^2\}$.

\subsection{The constraint submanifold $\M \subset T^*Q$ and the (reduced) differential space $\bar{\M}$}\label{Ss:Solids-Reduction}
 
 Given $(g,\veca)\in Q$, we denote by $({\bf M}, {\bf p})= (M_1, M_2, M_3, p_1, p_2)$ the coordinates on $T^*_{(g,\veca)} Q$ associated to the basis $\{ \vecL, \epsilon^1, \epsilon^2\}$.
 
Using the kinetic energy metric from \eqref{Eq:Lagrangian} we compute the 8 dimensional submanifold $\M= \kappa^\flat(D) \subset T^*Q$, which in this case is
\begin{equation}\label{Eq:Solids-M}
\M = \{(g,\veca; {\bf M}, {\bf p})\in T^*Q \ : \ p_1 =  m \langle \vecalpha, s\times \vecOm \rangle, \quad p_2 = m \langle \vecbeta , s \times \vecOm\rangle \},
\end{equation}
where ${\bf M} = \mathbb{I} \vecOm + m s\times (\vecOm \times s)$. We consider $(g,\veca,{\bf M})$ coordinates on $\M$ and the projection $\tau: \M\to Q$ is given by $\tau(g,\veca,{\bf M})=(g,\veca)$.

The lifted action of the Lie group $G=S^1\times E(2)$ to $T^*Q$ leaves the manifold $\M$ invariant. The induced $S^1$-action on $\M$ is given by 
$$
\left( h_\theta, (g,\veca,{\bf M}) \right) \mapsto (gh_\theta^{-1}, \bar{h}_\theta \veca , h_\theta {\bf M}),
$$ 
whereas the $E(2)$-action is given by 
$$
\left( (h_\varphi,(x,y)), (g,\veca,{\bf M} ) \right) \mapsto (h_\varphi g, \bar{h}_\varphi\veca + (x,y), {\bf M}).
$$

Since the action is proper, following Section \ref{Ss:Symmetries} and \cite{Book:BC}, the reduced space $\bar\M := \M/G$ is a differential space and, as usual, we denote $\rho:\M\to \bar\M$ the orbit projection.  To describe $\bar\M$ we will quotient the manifold $\M$ in two steps. First we consider the (free and proper) action of the normal subgroup $E(2)$ of $G$ on $\M$ and we see that $\M/E(2) \simeq S^2 \times \R^3$ with coordinates given by $(\vecgamma, {\bf M})$.  Then, we compute the quotient of the manifold $\M/E(2)$ by the remaining group $G/E(2) = S^1$. Following \cite{Book:CDS}, in order to describe the differential space $\bar\M$  we use {\it invariant theory}. The algebra of $S^1$-invariant polynomials on $S^2\times \R^3$ is generated by 
\begin{equation} \label{Eq:Solids-ReducedVar}
 \begin{split}
\tau_1 & = \gamma_3 \qquad \tau_2 = \gamma_1M_2 - \gamma_2M_1 \qquad \tau_3= \gamma_1M_1 + \gamma_2M_2, \\
\tau_4 & =M_3 \qquad \tau_5 = M^2_1 + M^2_2
 \end{split}
\end{equation}
with the relation $\tau_2^2 +\tau_3^2 = (1-\tau_1^2)\tau_5$ and $|\tau _3|\leq1$, $\tau_5 \geq 0$.  
Therefore, $\bar\M=\M/G$ is the semialgebraic variety in $\R^5$ defined by 
$$
\bar\M := \{(\tau_1, \tau_2,\tau_3, \tau_4, \tau_5) \ : \ \tau_2^2 +\tau_3^2 = (1-\tau_1^2)\tau_5, \ |\tau _3|\leq1, \ \tau_5 \geq 0\}.
$$
The singular points are $(\pm1,0,0,\tau_4,\tau_5)\in \bar\M$. Away from the singular points, $\bar\M$ is a 4-dimensional manifold. 

Recall that on the manifold $\M$ we have the nonholonomic bracket $\{\cdot, \cdot\}_\nh$ and its associated nonholonomic bivector $\pi_\nh$ as defined in \eqref{Eq:HamVectorField}.  Hence, there is a (singular) almost Poisson bracket $\{\cdot, \cdot \}_\red$ on the differential space $\bar\M$, as it was given in \eqref{Eq:ReducedBracket}.

Next, we verify the failure of the Jacobi identity for the reduced bracket $\{\cdot, \cdot \}_\red$ describing the (reduced) dynamics of the solids of revolution.

\subsection{The failure of the Jacobi identity of $\{\cdot, \cdot\}_\red$}
 
Following Prop.~\ref{P:Red-Jac} we start the section computing the 2-form $\langle \mathcal{J}, \mathcal{K}_\subW\rangle$ on $\M$ given in \eqref{Def:JK} in our chosen basis.

Recall that $(g, \veca, {\bf M})$ are our coordinates on the constraint submanifold $\M$. Then, we have the (local) basis of 1-forms on $\M$ given by  
\begin{equation}\label{Eq:basisM} 
\{\tau^*\vecL, \tau^*\epsilon^1, \tau^*\epsilon^2, d{\bf M} \},
\end{equation}
where $\tau:\M\to Q$ is the canonical projection and where $\epsilon^1$ and $\epsilon^2$ are the constraint 1-forms given in \eqref{Eq:basisM}.  The associated dual basis of vector fields on $\M$ is $\{ {\bf X}, \partial_{a_1}, \partial_{a_2}, \partial_{\bf M} \}$\footnote{We are denoting by  ${\bf X}, \partial_{a_1}, \partial_{a_2}$ vector fields on $Q$ and also vector fields on $\M$. When it is clear from the context, we will also denote by $\vecL$ and $\epsilon^1, \epsilon^2$ the 1-forms on $\M$ that are the pull backs of the corresponding 1-forms on $Q$} for $\partial_{\bf M} = (\partial_{M_1}, \partial_{M_2}, \partial_{M_3})$.

First we see that $\mathcal{A}_\subW: \M \to \g^*$ is given by $\mathcal{A}_\subM = \epsilon^1 \otimes (0;0, (1,0)) + \epsilon ^2 \otimes (0; 0 , (0,1))$ (where $\epsilon^1= \tau^*\epsilon^1$ and $\epsilon^2= \tau^*\epsilon^2$).
So,  the $\W$-curvature is
\begin{equation*}
 \begin{split}
  \mathcal{K}_\subW = & \ - d\langle \vecalpha, s \times \vecL\rangle \otimes (0;0, (1,0)) - d\langle \vecbeta, s\times \vecL \rangle \otimes (0; 0 , (0,1))\\
  =  & \ -  \left( \varrho\, \langle \vecbeta, d\vecL \rangle + \varrho' \beta_3 \langle \vecgamma, d\vecL\rangle - L' \beta_3 \, d\lambda_3\right)\otimes (0;0, (1,0)) \\ 
  & \ + \left( \varrho\langle \vecalpha, d\vecL \rangle + \varrho' \alpha_3 \langle \vecgamma, d\vecL\rangle - L' \alpha_3\,  d\lambda_3\right)\otimes (0; 0 , (0,1)), 
 \end{split}
\end{equation*}
where $d\vecL =  (d\lambda_1, d\lambda_2,d\lambda_3) = (\lambda_3\wedge \lambda_2, \lambda_1\wedge \lambda_3, \lambda_2\wedge \lambda_1)$ and $(\cdot) ' = \tfrac{d}{d\gamma_3}$.
On the other hand, from \eqref{Eq:Solids-M} we write the Liouville 1-form as 
$$
\Theta_\subM = \langle M, \vecL\rangle + m \langle \vecalpha, s\times \vecOm \rangle \epsilon^1 +  m \langle \vecbeta , s \times \vecOm\rangle \epsilon^2.
$$
Hence, the components of the canonical momentum map in $(0;0, (1,0))$ and $(0; 0 , (0,1)) \in \mathfrak{g}$ are 
$$
{\mathcal J}_1= {\bf i}_{\frac{\partial}{\partial a_1}} \Theta_\subM = m \langle \vecalpha, s\times \vecOm \rangle \quad \mbox{and} \quad {\mathcal J}_2 = {\bf i}_{\frac{\partial}{\partial a_2}} \Theta_\subM =  m \langle \vecbeta , s \times \vecOm\rangle.
$$
Following \eqref{Def:JK}, we obtain
\begin{lemma}\label{L:JK}
 For the body of revolution on a plane we have that  
\begin{equation}\label{Eq:Solids-JK}
 \begin{split}
  \langle \mathcal{J},\mathcal{K}_\subW\rangle  = & \  m\varrho \langle \vecgamma, s\rangle \langle \vecOm, d\vecL\rangle - m\varrho \langle \vecOm, \vecgamma\rangle \langle s, d\vecL\rangle - m \varrho' (\vecgamma \times (\vecOm \times s ) )_3 \langle \vecgamma, d\vecL\rangle \\
  & + m L' (\vecgamma \times (\vecOm \times s ) )_3 d\lambda_3,
 \end{split}
\end{equation}
where 
$(a \times b)_3$ denotes the third coordinate of the vector $a\times b$ for $a,b\in\R^3$, i.e., $(a \times b)_3 = a_1b_2 - a_2b_1$. 
\end{lemma}

Throughout this work, we will often use that $\langle \mathcal{J},\mathcal{K}_\subW\rangle  = \langle {\bf K}, d\vecL \rangle$ where ${\bf K} = (K_1, K_2, K_3)$ is given by ${\bf K} = - m\varrho \langle \vecgamma,s\rangle \vecOm + \mathcal{L}$ with 
\begin{equation}\label{Eq:L} 
{\mathcal L} = -m(\varrho^2 \langle \vecOm, \vecgamma \rangle  + \rho' c_3)\vecgamma +m(L\varrho \langle \vecOm, \vecgamma\rangle + L' c_3) {\bf e}_3,
\end{equation}
for $c_3 = (\vecgamma \times (\vecOm \times s ) )_3$ and ${\bf e}_3= (0,0,1)$.  Also, we observe that  
\begin{equation} \label{Eq:Q-P}
 {\mathcal L} = {\mathcal Q}\, \vecgamma + {\mathcal P}\, {\bf e}_3
\end{equation}
for ${\mathcal  Q}$ and ${\mathcal P} \in C^\infty(\M)^G$ given by ${\mathcal Q} = -m(\varrho^2 \langle \vecOm, \vecgamma \rangle  + \rho' c_3 )$ and $\mathcal{P} = m(L\varrho \langle \vecOm, \vecgamma\rangle + L' c_3)$.

\medskip

\begin{proposition}\label{P:notPoisson}
 The nonholonomic reduced bracket $\{\cdot, \cdot \}_{\emph\red}$ on $\M/G$ is not Poisson.
\end{proposition}

\begin{proof}
Using Prop.~\ref{P:Red-Jac}, we will prove that the right hand side of \eqref{Eq:Jac-pi_red} is different from zero for three functions on $\bar\M$. 
That is, we check that  $d \langle \mathcal{J},\mathcal{K}_\subW\rangle (Z_1,Z_2,Z_3)) \neq 0$ for $Z_1=\pi^\sharp_\nh(d\rho^*\tau_1)$,  $Z_3=\pi^\sharp_\nh(d\rho^*(\tau_3+ \tau_1 \tau_4))$ and $Z_4=\pi^\sharp_\nh(d\rho^*\tau_4)$.

Using that $\Omega_\C = (\lambda_i \wedge dM_i - M_i d\lambda_i - \langle \mathcal{J},\mathcal{K}_\subW\rangle ) |_\C,$
we obtain
$$
\pi_\nh^\sharp(\lambda_i) = \frac{\partial}{\partial M_i}, \quad \pi_\nh^\sharp(d{\bf M}) = -{\bf X} + ({\bf M} + {\bf K}) \times \frac{\partial}{\partial {\bf M}} \quad \mbox{and} \quad \pi_\nh^\sharp(\epsilon^1)= \pi_\nh^\sharp(\epsilon^2)=0.
$$ 
From \eqref{Eq:Solids-ReducedVar} $\rho^*\tau_1 = \gamma_3$,  $\rho^*(\tau_3 + \tau_1\tau_4)= \langle {\bf M}, \vecgamma \rangle$ and $\rho^*\tau_4 = M_3$, and then we compute
\begin{equation} 
 \begin{split}
Z_1 & =\pi^\sharp_\nh(d\gamma_3) = (\vecgamma\times \partial_{\bf M})_3,\\
Z_3 & =\pi^\sharp_\nh(d\langle {\bf M}, \vecgamma \rangle)= -\langle \vecgamma, {\bf X}\rangle + \langle \vecgamma \times {\bf K}, \partial_{\bf M}\rangle,\\
Z_4 & =\pi^\sharp_\nh(dM_3)= -X_3 + ({\bf M} \times \partial_{\bf M})_3 +  ({\bf K} \times \partial_{\bf M})_3.
 \end{split}
\end{equation}

Since $\langle \mathcal{J},\mathcal{K}_\subW\rangle = \langle {\bf K}, d\vecL\rangle$ then $d\langle \mathcal{J},\mathcal{K}_\subW\rangle =  \partial_{\gamma_j}{\bf K}_i \, d\gamma_j \wedge d\lambda_i + \partial_{M_j} {\bf K}_i \, dM_j \wedge d\lambda_i$.
We observe also that $\partial_{M_j} {\bf K}_i = m\varrho \langle \vecgamma, s\rangle \partial_{M_j}\Omega_i  + \partial_{M_j} {\mathcal L}_i$. Using \eqref{Eq:L} we see that 
\begin{equation}\label{Eq:CoordinatesLij} 
\partial_{M_j} {\mathcal L}_i = G_j \gamma_i + H_j\delta_{j3}, 
\end{equation}
 for $G_j = -m(\varrho^2 \partial_{M_j} \langle \vecOm, \vecgamma \rangle \vecgamma + \rho' \partial_{M_j} c_3)$ and $H_j = m(L\varrho \partial_{M_j}\langle \vecOm, \vecgamma\rangle - L'\partial_{M_j} c_3) $ and $\delta_{j3} = 1$ when $j=3$ and $0$ when $j=1,2$.
Finally, we see that
\begin{equation*}
 \begin{split}
d\langle \mathcal{J},\mathcal{K}_\subW\rangle (Z_1,Z_3,Z_4)  = &  -\left(\gamma_2\partial_{M_j}K_1 - \gamma_1 \partial_{M_j}K_2 \right) dM_j(Z_1) \\  
=  & - m \varrho\langle \vecgamma,s\rangle \left[ \gamma_2( -\gamma_2\partial_{ M_1}\Omega_1 + \gamma_1 \partial_{M_1}\Omega_2 ) - \gamma_1( -\gamma_2\partial_{M_2}\Omega_1 + \gamma_1 \partial_{M_2} \Omega_2) \right] \\ 
= & \   m \varrho\langle \vecgamma,s\rangle(\frac{\gamma^2_2+ \gamma_1^2}{\mathbb{I}_1+ m\langle s, s\rangle} )  \neq 0.
\end{split}
\end{equation*}
\end{proof}

\subsection{The dynamical gauge transformation}\label{Ss:Solids-Gauge}

Following Section \ref{Ss:Gauge} we propose a 2-form $B$ defining a {\it dynamical gauge transformation} in order to obtain a new bivector field describing the nonholonomic dynamics in the sense that 
\begin{equation}\label{Eq:NHdynamics}
- \pi_\B^\sharp(d\Ham_\subM) = X_\nh = \langle \vecOm, {\bf X}\rangle + \langle \dot {\bf M }, \frac{\partial}{\partial {\bf M}}\rangle,
\end{equation}
where $\dot {\bf M} = {\bf M} \times \vecOm + m (\vecOm \times \dot s )\times s - m\mbox{g} s\times \vecgamma$ (see \cite[Ch.6]{Book:CDS} and \cite{BorisovMamaev2002}). 

From the computation of the 2-form $\langle \mathcal{J}, \mathcal{K}_\subW\rangle$ in Lemma \ref{L:JK}, let us consider its first term, given by
\begin{equation}\label{Eq:DynGauge}
 B =  m\varrho \langle \vecgamma, s\rangle \langle \vecOm, d\vecL\rangle.
\end{equation}

\begin{proposition}
 The 2-form $B$ is a dynamical gauge transformation of $\pi_{\emph\nh}$ on $\M$.
\end{proposition}

\begin{proof} Following Definition \ref{Def:Gauge}, we check that conditions $(i)$ and $(ii)$ are satisfied.   
 First observe $B$ is semi-basic with respect to the bundle $\tau:\M\to Q$ and therefore $(\Omega_\subM +B)|_\C$ is nondegenerate (see Remark~\ref{R:SemiBasic}). Second, using \eqref{Eq:NHdynamics}, we see that $${\bf i}_{X_\nh} B  = m\varrho \langle \vecgamma, s \rangle \langle \vecOm, \vecOm \times \vecL\rangle  = 0.$$ So, $B$ given in \eqref{Eq:DynGauge} defines a dynamical gauge transformation. 
\end{proof}

It follows that the gauge transformation of $\pi_\nh$ by the 2-form $B$ produces a new bivector field $\pi_\B$ in the sense of \eqref{Eq:B-NHbracket} that describes the dynamics as in \eqref{Eq:NHdynamics}.  The new bivector field $\pi_\B$ is determined by the 2-section $(\Omega_\subM +B)|_\C= (\lambda_i \wedge dM_i - M_i d\lambda_i - \langle \mathcal{L},d\vecL \rangle ) |_\C$ for $\mathcal{L}$ given in \eqref{Eq:L} and where we used that 
\begin{equation}\label{Eq:Solids-JK-B}
\langle \mathcal{J},\mathcal{K}_\subW\rangle - B = \langle{\mathcal L},d\vecL\rangle = {\mathcal Q}\langle \vecgamma, d\vecL\rangle + {\mathcal P} d\lambda_3.
\end{equation} 
For completeness we write the new bivector $\pi_\B$ describing the motion of any solid of revolution rolling without sliding on a plane: 
$$
\pi_\B = {\bf X}_i \wedge \frac{\partial}{\partial M_i} - (M_1+ \mathcal{L}_1)\frac{\partial}{\partial M_2}\wedge \frac{\partial}{\partial M_3} - (M_2+\mathcal{L}_2)\frac{\partial}{\partial M_3}\wedge \frac{\partial}{\partial M_1} - (M_3+\mathcal{L}_3)\frac{\partial}{\partial M_1}\wedge \frac{\partial}{\partial M_2}.
$$

We denote by $\{\cdot, \cdot \}_\B$ the associated almost Poisson bracket.  
Since the 2-form $B$ is $G$-invariant, the bracket $\{\cdot, \cdot \}_\B$  is also $G$-invariant and thus there is an induced (singular) almost Poisson bracket $\{\cdot, \cdot \}_\red^\B$ on $\bar\M$ as in \eqref{Eq:GaugeReducedBracket}. That is, for $f,g \in C^\infty(\bar\M)$ we have 
\begin{equation}\label{Eq:Solids-GaugeReducedBracket}
\{f,g\}_\red^\B \circ \rho = \{\rho^*f,\rho^*g\}_\B.
\end{equation}
where, as usual,  we are identifying the smooth functions on $\bar\M$ with $C^\infty(\M)^G$.

\bigskip

As we will see, this dynamical gauge transformation has the property that the reduced bracket $\{\cdot, \cdot\}_\red^\B$ is Poisson, i.e., it hamiltonizes the system. We will check that in two ways: in the next Theorem we prove it using a direct argument based on the formulation \eqref{Eq:Jac-pi_Bred}.  
In the next section, we will derive this result from the study of horizontal gauge momenta associated to the nonholonomic system. This second approach also clarifies the seemingly ``ad-doc'' choice of $B$.

\begin{theorem}\label{T:Solids-Poisson}
 The reduced bracket $\{\cdot, \cdot\}_{\emph\red}^\B$ on $\bar\M$, describing the dynamics of a solid of revolution on a plane, is Poisson. 
\end{theorem}

\begin{proof}

To prove the theorem, we use formula \eqref{Eq:Jac-pi_Bred}, that is, we show that 
$$
(d\langle \mathcal{J},\mathcal{K}_\subW\rangle -dB)|_{U_\B} = 0.
$$   
We will see that 
$d[\langle \mathcal{J},\mathcal{K}_\subW\rangle -B](Z_k,Z_l,Z_h) = 0$ for $k,j,h =1,...,5,$ where 
\begin{equation*} 
 \begin{split}
 Z_1 = & \ \pi_\B^\sharp(d\rho^*\tau_1) = \pi_\B^\sharp(d\gamma_3) = (\vecgamma\times \partial_{\bf M})_3,\\ 
 Z_2 = & \ \pi_\B^\sharp(d\rho^*\tau_2) = \pi_\B^\sharp(d\langle {\bf M}\times \vecgamma , {\bf e}_3 \rangle)= -\gamma_2X_1 + \gamma_1X_2  + [M_1\gamma_3 -\gamma_1 (M_3 + \mathcal{Q} \gamma_3 + \mathcal{P})] \partial_{M_1} \\
 & \ \ \ \ \ \ \ \ \ \ \ \ \ \ \ \ \ \ \ \ \ \ \ \ \ \ \ \ \ \ \ \ \ \ \ \ \ \ \ \ \ \ \ \ + [M_2\gamma_3 -\gamma_2 (M_3 + \mathcal{Q} \gamma_3 + \mathcal{P})] \partial_{M_2} + \mathcal{Q}(\gamma_1^2+\gamma_2^2) \partial_{M_3},\\
 Z_3  = & \ \pi_\B^\sharp(d\rho^*(\tau_3+\tau_1\tau_4)) = \pi_\B^\sharp(d\langle {\bf M}, \vecgamma \rangle)= -\langle \vecgamma, {\bf X}\rangle  - {\mathcal P}(\vecgamma\times \partial_{\bf M})_3 ,\\
 Z_4 = & \ \pi_\B^\sharp(d\rho^*\tau_4) = \pi_\B^\sharp(dM_3)= -X_3 +(M_1\partial_{M_2}-M_2\partial_{M_1} ) + {\mathcal Q}(\gamma_1\partial_{M_2}-\gamma_2\partial_{M_1}),\\
 Z_5 = & \ \tfrac{1}{2}\pi_\B^\sharp(d\rho^*\tau_5) = \pi_\B^\sharp(M_1dM_1 + M_2dM_2) = -M_1X_1 - M_2X_2 +(M_3+\mathcal{L}_3)(M_2\partial_{M_1} - M_1\partial_{M_2}) \\ 
 & \ \ \ \ \ \ \ \ \ \ \ \ \ \ \ \ \ \ \ \ \ \ \ \ \ \ \ \ \ \ \ \ \ \ \ \ \ \ \ \ \ \ \ \  \ \ \ \ \ \ \ \  + {\mathcal Q} (M_1\gamma_2-M_2\gamma_1)\partial_{M_3},
\end{split}
\end{equation*}
recalling that $\mathcal{L} = \mathcal{Q}\, \vecgamma + {\mathcal P}\, {\bf e}_3$ as in \eqref{Eq:Q-P}.
Using \eqref{Eq:Solids-JK-B} we compute
$$
d\langle \mathcal{J},\mathcal{K}_\subW\rangle +dB = \frac{\partial{\mathcal L}_i}{\partial M_j} \, (dK_i \wedge d\lambda_j) - \left( \gamma_2 \frac{\partial \mathcal{P}}{\partial \gamma_1} - \gamma_1 \frac{\partial \mathcal{P}}{\partial \gamma_2} \right) \lambda_1\wedge \lambda_2 \wedge \lambda_3.
$$ 
Following \eqref{Eq:CoordinatesLij} and the fact the functions $G_j$ and $H_j$ can be written as 
$$
G_j = G(\gamma_3) A_j^{-1}s_j + g(\gamma_3) \delta_{j3} \quad \mbox{and} \quad H_j = H(\gamma_3) A_j^{-1} s_j + h(\gamma_3) \delta_{j3},
$$
for $A_j = \mathbb{I}_j + m \langle s, s \rangle$, it is straightforward to check that 
\begin{equation}\label{Eq:Solids-ProofZ1}
{\bf i}_{Z_1} [d\langle \mathcal{J},\mathcal{K}_\subW\rangle +dB] = 0.
\end{equation}
Then, we also see that ${\bf i}_{Z_3, Z_4} [d\langle \mathcal{J},\mathcal{K}_\subW\rangle -dB]= 0$. With a little more work one checks that $[d\langle \mathcal{J},\mathcal{K}_\subW\rangle -dB](Z_2,Z_3, Z_5)  = 0$, but then it is easy to see that also $ [d\langle \mathcal{J},\mathcal{K}_\subW\rangle -dB](Z_2,Z_4, Z_5) = 0$. 
\end{proof}

From Theorem \ref{T:Solids-Poisson} we conclude that the dynamics of the solids of revolution rolling on a plane without sliding is hamiltonizable through a reduction process. 

Away from the singularities of $\bar\M$, one has a Poisson bracket $\{\cdot , \cdot \}_\red^\B$ of rank 2 (i.e., 2-dimensional symplectic leaves).  Again, for completeness we write the reduced bivector field $\pi_\red^\B$ away the singularities of the space $\bar\M$: 
\begin{equation*}
 \begin{split}
\pi_\red^\B = & \  (1-\tau_1^2) \left(\, \partial_{\tau_1}\wedge \partial_{\tau_2} - (\tau_4 + \mathcal{L}_3) \partial_{\tau_2} \wedge \partial_{\tau_3} + \mathcal{Q} \partial_{\tau_2} \wedge \partial_{\tau_4} \, \right) \\
&  + \tau_2 \left(\,  \partial_{\tau_1} \wedge \partial_{\tau_5} -(\tau_4+\mathcal{L}_3)\partial_{\tau_3}\wedge \partial_{\tau_5} + 2\mathcal{Q}\partial_{\tau_4} \wedge \partial_{\tau_5} \, \right)  + 2(\tau_1\tau_5 -\tau_3(\tau_4+\mathcal{L}_3)) \partial_{\tau_2} \wedge \partial_{\tau_5}.
 \end{split}
\end{equation*}

\bigskip

In the next section we will see a more justified way to choose the dynamical gauge transformation \eqref{Eq:DynGauge} so that $\{\cdot, \cdot\}_\red^\B$ is Poisson.  The choice follows  Theorem~\ref{T:Casimirs} based on the existence of conserved quantities (the horizontal gauge momenta) of the system.  As a consequence, we will have a more direct proof of the fact that $\{\cdot, \cdot\}_\red^\B$ is Poisson. 

\subsection{Conserved quantities} 

Following \cite{Book:CDS}, the nonholonomic system describing the dynamics of a solid of revolution admits two additional constants of motion that are, in fact, horizontal gauge momenta with respect to the $G$-action in the sense of Def.~\eqref{Def:Gauge-Momentum}, \cite{FSG2008}. In what follows, we study the horizontal gauge momenta of $X_\nh$ using Section \ref{Ss:Symm-Conservation}.

More precisely, following  \eqref{Eq:gs} we have that $\{\xi_1, \xi_2\}$ is a basis of sections of $\Gamma(\mathfrak{g}_\subSS)$.  
Recall that $P_{\g_\subSS}:\g\times Q \to \g_\subSS$ is the projection associated to the splitting \eqref{Eq:gS+gW} and observe that 
\begin{equation}\label{Eq:Proj-gs} 
\xi_1 = P_{\mathfrak{g}_\subSS} (\eta_1 )\qquad \mbox{and} \qquad \xi_2 = P_{\mathfrak{g}_\subSS} ( \eta_2 ),
\end{equation}
where $\eta_1 = (1; 0, {\bf 0})$ and $\eta_2 = (0; 1, {\bf 0}) $ are elements of the Lie algebra $\mathfrak{g}$ (i.e., constant sections of the bundle $\mathfrak{g}\times Q \to Q$). 
Their infinitesimal generators with respect to the $G$-action on $\M$  are vector fields with values in $\S$ given by
\begin{equation} \label{Eq:InfinitesGenerators-M}
{(\xi_1)}_\subM = -X_3 - (M_2 \frac{\partial}{\partial M_1} - M_1 \frac{\partial}{\partial M_2}) \qquad \mbox{and}  \qquad {(\xi_2)}_\subM = \langle \vecgamma, {\bf X} \rangle.
\end{equation}

From Prop.~\ref{Prop:OmegaJK-Moment} and \eqref{Def:j}, we define the functions $j_1$ and $j_2$ in $C^\infty(\M)$ by 
$$
j_1 = \langle {\mathcal J}^\nh, \xi_1 \rangle = {\bf i}_{(\xi_1)_\subM} \Theta_\subM =  - M_3 \qquad \mbox{and}\qquad j_2 = \langle {\mathcal J}^\nh, \xi_2 \rangle = {\bf i}_{(\xi_2)_\subM} \Theta_\subM = \langle \vecgamma, {\bf M}\rangle.
$$

The functions $j_1$ and $j_2$ are {\it not} first integrals of the dynamics of the solids of revolutions on a plane.  In fact, using Prop.~\ref{P:Obstruction} and \eqref{Eq:NHdynamics}  we have that
\begin{equation}\label{Eq:JK(xi,Xnh)} 
\begin{split}
X_\nh(j_1) = & \ \langle \mathcal{J}, \mathcal{K}_\subW\rangle ((\xi_1)_\subM, X_\nh)= \langle \mathcal{J}, \mathcal{K}_\subW\rangle (-X_3, \langle\vecOm,{\bf X}\rangle)= -\mathcal{Q}\,A_1^{-1} \tau_2,\\
X_\nh(j_2) = & \ \langle \mathcal{J}, \mathcal{K}_\subW\rangle ((\xi_2)_\subM, X_\nh)= \langle \mathcal{J}, \mathcal{K}_\subW\rangle (\langle \vecgamma,{\bf X}\rangle , \langle\vecOm,{\bf X}\rangle)= - {\mathcal P} \, A_1^{-1} \tau_2.
\end{split}
\end{equation}
where $A_1 = \mathbb{I}_1 + m\langle s,s\rangle.$

\begin{lemma}\label{L:Solids-G-Functions} The functions $j_1$, $j_2$ and   ${\mathcal Q}$, ${\mathcal P} \in C^\infty(\M)$ are $G$-invariant and they can be written in the reduced space $\bar\M$ as linear functions in the variables $\tau_3$ and $\tau_4$:
\begin{enumerate}
 \item[$(i)$] The functions $j_1$ and $j_2 \in C^\infty(\bar\M)$ verify that 
$\left(\!\! \begin{array}{c}  j_1 \\ j_2\end{array}\!\!\right) = \left( \begin{array}{cc}  0 & -1 \\ 1 & \tau_1 \end{array}\right) \left(\!\! \begin{array}{c}  \tau_3 \\ \tau_4\end{array}\!\!\right)$
\item[$(ii)$] The functions ${\mathcal Q}$ and ${\mathcal P} \in C^\infty(\bar\M)$ verify that $\left(\!\!\begin{array}{c} {\mathcal Q} \\  {\mathcal P} \end{array}\!\!\right) =  [{\mathcal{QP}}]   \left(\!\!\begin{array}{c}  \tau_3 \\ \tau_4 \end{array}\!\!\right)$ for 
$$
[{\mathcal{QP}}] = [{\mathcal{QP}}](\tau_1) = \frac{m}{P(\tau_1)} \left( \begin{array}{cc} 
 \mathbb{I}_3(- \varrho^2 + \varrho' \zeta)  - \varrho^3 \sigma  & - \varrho\mathbb{I}_1(\varrho \tau_1 + \varrho' (1-\tau_1^2)) - \varrho^2 \zeta \sigma \\ 
 \mathbb{I}_3(\varrho L - L' \zeta) + \varrho^2L\sigma &  \varrho\mathbb{I}_1( L \tau_1 + L' (1-\tau_1^2))+ \varrho\zeta L \sigma  \end{array} \right)
$$
with $\sigma = m \langle s,\vecgamma\rangle = m( \varrho(1-\tau_1^2)+\zeta \tau_1)$ and $P(\tau_1) = \mathbb{I}_1\mathbb{I}_3 + m\langle \mathbb{I}s,s\rangle$.
\end{enumerate}
\end{lemma}

\begin{proof}
 Item $(i)$ of the Lemma is straightforward. For item $(ii)$ we used the fact that 
 \begin{equation} \label{Eq:Solids:Om}
 \vecOm = A^{-1} {\bf M} + m\langle s, \vecOm\rangle A^{-1} s,
 \end{equation}
 for $A$ the $3\times 3$ diagonal matrix given by $A = \mathbb{I} + m \langle s,s\rangle \textup{Id}$ \, and $\langle s,\vecOm\rangle = \frac{\langle A^{-1} {\bf M}, s\rangle}{E}$, where $E=E(\tau_1) = 1 - m \langle A^{-1}s,s\rangle \neq 0$. If we call $\sigma_1 = \gamma_1\Omega_2 - \gamma_2\Omega_1$ and $\sigma_2 = \Omega_3$ we obtain that 
 $$
 B \left(\!\! \begin{array}{c} \sigma_1\\ \sigma_2\end{array}\!\!\right)= \left(\!\! \begin{array}{c}  \mathcal{Q} \\ \mathcal{P} \end{array}\!\!\right)\quad \mbox{and} \quad \frac{1}{P(\tau_1)} C\left(\!\! \begin{array}{c}  \tau_3 \\ \tau_4\end{array}\!\!\right) = \left(\!\! \begin{array}{c}  \sigma_1 \\ \sigma_2\end{array}\!\!\right),
 $$ 
where $B$ and $C$ are the matrices 
 $$B = m \left(\!\begin{array}{cc} 
 - \varrho^2 + \varrho' \zeta \ \ \  & - \varrho^2 \tau_1 - \varrho'\varrho (1-\tau_1^2) \\ \varrho L - L' \zeta &  \varrho( L \tau_1 + L' (1-\tau_1^2)) \end{array} \! \right)\quad \mbox{and} \quad C =
\left(\!\begin{array}{cc} \mathbb{I}_3  + m\varrho^2(1-\tau_1^2) &  m\varrho\zeta(1-\tau_1^2)  \\  m\varrho \zeta & \mathbb{I}_1 + m \zeta^2 \end{array}\!\right).$$
Hence,
$[{\mathcal{QP}}] = \frac{1}{P(\tau_1)}  B.C$. 
 
\end{proof}

Following Prop.~\ref{P:gaugeMomentum}, the momentum equation \eqref{Eq:DifEq} can be written as a linear system of ordinary differential equations of first order, as the following theorem shows. 

\begin{theorem} \label{T:Solids-Diff}
 The nonholonomic vector field $X_{\emph\nh}$ admits two gauge momenta $J_1$ and $J_2$ that are a combination of the functions $j_1$ and $j_2$:   
\begin{equation}\label{Eq:Solids-ConserQuant}
 J_1 = f_1(\tau_1)  j_1 + g_1(\tau_1)  j_2  \quad \mbox{and} \quad J_2 = f_2(\tau_1)  j_1 + g_2(\tau_1)  j_2,
\end{equation}
where $f_i$ and $g_i \in C^\infty(Q)^G$ satisfy the following system of first order linear differential equations
\begin{equation} \label{Eq:Solids-Diff} 
\left(\!\! \begin{array}{c} f' \\ g'\end{array}\!\! \right) = \left( \begin{array}{cc}  \tau_1 & -1 \\ 1 & 0  \end{array}\right) [\mathcal{QP}]^T \left(\!\! \begin{array}{c} f\\ g\end{array}\!\! \right).
\end{equation} 
 for $(\cdot) ' = \tfrac{d}{d\tau_1}(\cdot)$. Here $[\mathcal{QP}]^T$ denotes the transpose of the matrix $[\mathcal{QP}]$.
\end{theorem}

\begin{proof}
We start by writing the differential equation \eqref{Eq:DifEq} in the case of the solid of revolution: if $J=fj_1 + gj_2$ is a horizontal gauge momentum, then $f,g\in C^\infty(Q)$ satisfy
\begin{equation}\label{Eq:Solids-DifEq}
 j_1 df(X_\nh)+j_2dg(X_\nh) + f\langle \mathcal{J}, \mathcal{K}_\subW\rangle ((\xi_1)_\subM, X_\nh) + g\langle \mathcal{J}, \mathcal{K}_\subW\rangle ((\xi_2)_\subM, X_\nh) = 0.
 \end{equation}
First we see that if $f$ and $g$ are $G$-invariant functions on $Q$, we have that $df(X_\nh) = f' \rho^*d\tau_1(X_\nh)$ and $dg(X_\nh) = g' \rho^*d\tau_1(X_\nh)$ for $(\cdot) ' = \tfrac{d}{d\tau_1}(\cdot)$. Then, using that $\rho^*d\tau_1 = \gamma_1\lambda_2 - \gamma_2\lambda_1$ and equations \eqref{Eq:NHdynamics} with \eqref{Eq:Solids:Om} we have that 
$
df(X_\nh) = f' A_1^{-1} \tau_2$ and $dg(X_\nh) = g' A^{-1}_1\tau_2.$
Second, using \eqref{Eq:JK(xi,Xnh)}, the differential equation \eqref{Eq:Solids-DifEq} reads   
 \begin{equation}\label{Eq:Solids-Diff-Proof1} 
j_1 f'+ j_2 g'- f\mathcal{Q} - g \mathcal{P} = 0.  
 \end{equation} 
Finally, by Lemma \ref{L:Solids-G-Functions} $(i)$ we can write \eqref{Eq:Solids-Diff-Proof1} as
$ (\tau _3, \tau_4) \left(\!\! \begin{array}{cc} 0 & 1 \\ -1 & \tau_1  \end{array}\!\!\right) \! \left(\!\!\begin{array}{c} f' \\ g'\end{array}\!\!\right) -  (\tau_3, \tau_4) [\mathcal{QP}]^{T}\! \left(\!\!\begin{array}{c} f\\ g\end{array}\!\!\right)=0$, and
so \eqref{Eq:Solids-Diff} follows. The system of first order linear differential equations
\eqref{Eq:Solids-Diff} admits two solutions $(f_1,g_1)$ and $(f_2,g_2)$ and thus we obtain \eqref{Eq:Solids-ConserQuant}.

\end{proof}

\begin{corollary} The conserved quantities $J_1$ and $J_2 \in C^\infty(\M)^G$  of $X_{\emph\nh}$ are of the form
 $$
 J_i = \langle {\mathcal J}^{\emph\nh}, \chi_i\rangle  \ \ \ \mbox{for } i=1,2,
 $$
 where $\chi_i = f_i \xi_1 + g_i \xi_2 \in \Gamma(\mathfrak{g}_\subSS)$ for $\xi_1 = P_{\mathfrak{g}_\subSS} (( 1; 0, {\bf 0}))$ and $\xi_2 = P_{\mathfrak{g}_\subSS} ( (0; 1, {\bf 0}))$ and with the pairs $(f_i,g_i)$ satisfying \eqref{Eq:Solids-Diff}.  In other words,  each $J_i$ is a horizontal gauge momentum with $\chi_i$ the associated horizontal gauge symmetry. 
\end{corollary}

The functions $J_1$ and $J_2$ are the known first integrals of the solids of revolutions, found first in \cite{Book:CDS}. In what follows we see that they play a fundamental role for the reduced bracket $\{\cdot, \cdot \}_\red^\B$.

Now, we consider the ($G$-invariant) 2-form $B$ given in \eqref{Eq:DynGauge} and the dynamically gauge-related bivector field $\pi_\B$ to the nonholonomic bivector $\pi_\nh$. Following Theorem~\ref{T:Casimirs} we show that the reduced bracket  $\{\cdot, \cdot\}_{\red}^\B$ has {\it two} (independent) Casimirs induced by the horizontal gauge momenta $J_1$ and $J_2$.

\begin{theorem}\label{T:Solids-Casimirs}
 The $G$-invariant functions $J_1$ and $J_2$ induce the Casimirs $\bar{J}_1$ and $\bar{J}_2 \in C^\infty(\bar\M)$ of the reduced bracket $\{\cdot, \cdot\}_{\emph\red}^\B$, where $\bar{J}_1$ and $\bar{J}_2$ satisfy that $J_1 = \rho^*\bar{J}_1$ and $J_2= \rho^*J_2$. 
\end{theorem}

\begin{proof}  First we will see that, for each $i=1,2$, if $\chi_i = f_i\xi_1 +g_i\xi_2$ and $J_i = f_ij_1+g_ij_2$, then  $\pi_\B^\sharp(dJ_i)= -(\chi_i)_\subM$ when the pair $(f_i, g_i)$ satisfies  \eqref{Eq:Solids-Diff}. 
Following Thm.~\ref{T:Casimirs}, it is sufficient to prove that ${\bf i}_{(\chi_i)_\subM} B = \Lambda_i$ for $\Lambda_i = {\bf i}_{(\chi_i)_\subM} \langle \mathcal{J}, \mathcal{K}_\subW\rangle + j_1 df_i+  j_2dg_i. $

Observe that $df_i = f'_i (\gamma_1\lambda_2 - \gamma_2\lambda_1)$ and $dg_i = g'_i (\gamma_1\lambda_2 - \gamma_2\lambda_1)$. Using Lemma \ref{L:JK}, we see that $\langle \mathcal{J}, \mathcal{K}_\subW\rangle = B + \langle \mathcal{L}, d\vecL\rangle$.  Therefore, by \eqref{Eq:Q-P} we obtain
\begin{equation}
 \begin{split}
  \Lambda_i & = {\bf i}_{(\chi_i)_\subM} B + f_i {\bf i}_{(\chi_i)_\subM} \langle \mathcal{L},d\vecL\rangle +  g_ i{\bf i}_{(\chi_i)_\subM} \langle \mathcal{L},d\vecL\rangle +  j_1 df_i+  j_2dg_i \\
  & = {\bf i}_{(\chi_i)_\subM} B - ( f_i{\mathcal Q} + g_i \mathcal{P} -  j_1 f'_i -  j_2g'_i )(\gamma_1\lambda_2 - \gamma_2\lambda_1) = {\bf i}_{(\chi_i)_\subM} B,
 \end{split}
\end{equation}
where in the last equation we used that $f_i$ and $g_i$ satisfy equation \eqref{Eq:Solids-Diff-Proof1} since they are the coefficients of the gauge momenta $J_1$ and $J_2$. Therefore, we have that $\pi_\B^\sharp(dJ_i) = -(\chi_i)_\subM$ and so since $J_i\in C^\infty(\M)^G$, we conclude that $\{\bar{J}_i, \cdot \}_{\red}^\B=0$ for $\bar{J}_i\in C^\infty(\bar\M)$ such that $\rho^*\bar{J}_i = \bar{J}_i$.    

\end{proof}

So, Theorem \ref{T:Solids-Casimirs} gives an alternative viewpoint and proof of Theorem \ref{T:Solids-Poisson}.

\begin{corollary} \label{C:Solids-Poisson}
 The reduced bracket $\{\cdot, \cdot\}_{\emph\red}^\B$ on $\M/G$ is Poisson.  
\end{corollary}

\begin{proof}
 From Lemma \ref{L:Solids-G-Functions}$(i)$ we observe that $J_1, J_2$ depend linearly from the $G$-invariant functions $\tau_3$ and $\tau_4$.  Therefore, from \eqref{Eq:Solids-ReducedVar}, the $G$-invariant functions on $\M$ are also generated by $J_1, J_2, \tau_1, \tau_2,\tau_5$. Since $\bar{J}_1, \bar{J}_2$ are Casimirs of $\{\cdot, \cdot\}_\red^\B$ we have to check the Jacobi identity only on $\tau_1, \tau_2,\tau_5$.  It is straightforward to check that $cyclic [ \{ \{\tau_1, \cdot\}_\red^\B, \cdot\}_\red^\B ]=0$  using \eqref{Eq:Solids-ProofZ1}.
 
\end{proof}

\begin{remark} Away from the singularities (where $\bar\M$ is a 4-dimensional manifold) we can argue that two Casimirs already guarantee that $\{\cdot, \cdot\}_\red^\B$ is Poisson.  The argument goes as follows, with 2 Casimirs $\bar{J}_1, \bar{J}_2$, it is clear that the characteristic distribution of $\{\cdot, \cdot\}_\red^\B$ has rank 2. Moreover, the characteristic distribution is integrable since its annihilator is given by the exact forms $\{d\bar{J}_1, d\bar{J}_2\}$ and hence, we see that the bracket admits a 2-dimensional foliation. Finally, on each leaf there is a 2-form that is closed since the leaves are 2-dimensional. 

\end{remark}

\begin{remark} \label{R:Solids-not-Poisson} None of the conserved quantities $J_i$, for $i=1,2$  are horizontal gauge momenta with respect to the $E(2)$-action. That is, the vector fields $(\chi_i)_\subM$ are {\it not} infinitesimal generators with respect to the $E(2)$-action. Observe that if one of the conserved quantities, say $J_1$, was a horizontal gauge momentum with respect to the $E(2)$-action then $f_1$ should be zero and thus $J_1$ and $j_2=\langle {\bf M}, \vecgamma\rangle$ would be linearly dependent.  It is easy to see that the pair $(0, g_1)$ is not a solution of \eqref{Eq:Solids-Diff}.  Therefore, there is {\it no} bracket (describing the dynamics) on the manifold $\M/E(2)$ having $\bar{J}_1$ or $\bar{J}_2$ as Casimir functions (c.f. \cite{Bolsinov}).
\end{remark}

\section{Particular examples: the Routh sphere and the rolling ellipsoid}

\subsection{The Routh sphere} \label{S:Routh}

\subsubsection{Preliminaries and hamiltonization}

Consider a sphere of radio $r$ with its geometric center not coinciding with the center of mass.  In the plane perpendicular to the line joining the center of mass and the geometric center, the inertia tensor of the sphere has two equal principal moments of inertia. 
Following \cite{Cushman1998} we study the dynamics of this sphere --called the Routh sphere-- rolling on a plane without sliding, see also \cite{Bizy-Tsiganov}, \cite{routh}.  

As in Section \ref{S:Body-of-revolution}, the coordinates describing the position and velocities of the sphere are $((g,\veca ), \Omega, \dot \veca)$ on $T(SO(3)\times \R^3)$, where $\Omega = g^{-1} \dot g$ with the lagrangian and the nonholonomic constraints given in \eqref{Eq:Lagrangian} and \eqref{Eq:Constraints}, respectively.   In this particular case we have that  $$
s = -r \vecgamma + le_3,
$$
where $l$ is the distance between the center of mass and the geometric center of the sphere.  Therefore, following \eqref{Eq:s} we have $\varrho(\gamma_3) = -r$, $\zeta(\gamma_3) = -r \gamma_3 +l$ and $L(\gamma_3)= -l$, and consequently $\varrho'=L' =0$.

From Lemma \ref{L:JK} (see \eqref{Eq:Solids-JK}) we obtain that
\begin{equation}
  \langle {\mathcal J}, \mathcal{K}_\subW \rangle = -rm \langle \vecgamma, s\rangle \langle \vecOm, d\vecL\rangle + r m \langle \vecOm, \vecgamma\rangle \langle s, d\vecL\rangle, 
\end{equation}
and following \eqref{Eq:DynGauge} we have also that 
\begin{equation}\label{Eq:RouthGauge} 
B= -rm \langle \vecgamma, s\rangle \langle \vecOm, d\vecL\rangle
\end{equation}
defines a dynamical gauge transformation of $\pi_\nh$. We learned about this 2-form $B$ from Luis Garcia-Naranjo during a talk in 2015; it appears in \cite{LuisJames} with a different approach.

Following Section \ref{Ss:Solids-Gauge}, the gauge transformation of $\pi_\nh$ by the 2-form $B$ above, induces a new $G$-invariant bivector field $\pi_\B$. The bivector $\pi_\B$ induces a reduced Poisson bracket $\{\cdot,\cdot\}_\red^\B$ on $\bar\M$ describing the reduced dynamics (Thm.~\ref{T:Solids-Poisson}).

Observe that 
$$ 
\langle {\mathcal J}, \mathcal{K}_\subW \rangle - B = r m \langle \vecOm, \vecgamma\rangle \langle s, d\vecL\rangle = -r^2 m \langle \vecOm, \vecgamma\rangle \langle \vecgamma, d\vecL\rangle + r l m \langle \vecOm, \vecgamma\rangle d\lambda_3,
$$  
and thus we get that $\mathcal{Q} = -r^2 m \langle \vecOm, \vecgamma\rangle$ and $\mathcal{P} = r l m \langle \vecOm, \vecgamma\rangle$.

\begin{remark} Following the idea of Remark~\ref{R:Solids-not-Poisson}, we could wonder if it is possible to obtain a twisted Poisson on $\M/E(2)$; that is, if the reduction by $E(2)$ of $\{\cdot, \cdot\}_\B$ admits a foliations of almost symplectic leaves (as in the case of the Chaplygin ball, Example \ref{Ex:ChapBall}).  In order to obtain a bracket with such properties, the 2-form $\langle {\mathcal J}, \mathcal{K}_\subW \rangle - B$ has to be basic with respect to the bundle $\M\to \M/E(2)$, see item $(iii)$ in Sec.~\ref{Ss:Regular}. On one hand, we see that the first term $\mathcal{Q}\langle \vecgamma, d\vecL\rangle$ is basic. However, on the other hand, the second term $\mathcal{P}d\lambda_3$  is not basic. Moreover since ${\bf i}_{X_\nh} d\lambda_3 \neq 0$, that term cannot be considered as part of the {\it dynamical gauge transformation}.  In conclusion, we cannot expect to have a twisted Poisson bracket on $\M/E(2)$. 
 
\end{remark}

Now, following Section \ref{Ss:Solids-Reduction} the action of the Lie group $G= S^1\times SE(2)$ induces the reduced (differential space) $\bar\M = \M/ G$ and the reduced bracket $\{\cdot, \cdot \}_\red^\B$ defined on functions $C^\infty(\bar\M)$ as in \eqref{Eq:Solids-GaugeReducedBracket}.  From Theorem \ref{T:Solids-Poisson} we have that the reduced bracket $\{\cdot, \cdot \}_\red^\B$ is Poisson on $\bar\M$.  In the next section, we will analyze the conserved quantities and argue that the reduced bracket $\{\cdot, \cdot \}_\red^\B$ is Poisson by following Theorem \ref{T:Solids-Casimirs} and Corollary \ref{C:Solids-Poisson}.

\subsubsection{Conserved quantities and Casimirs}\label{Ss:Casimir}

To compute the conserved quantities that are gauge momenta, we use Thm.~\ref{T:Solids-Diff} and we see that 
\begin{equation*}
 J_1 = f_1(\gamma_3)  j_1 + g_1(\gamma_2)  j_2  \quad \mbox{and} \quad J_2 = f_2(\gamma_3)  j_1 + g_2(\gamma_2)  j_2,
\end{equation*}
for $f_i$ and $g_i \in C^\infty(Q)^G$  satisfying \eqref{Eq:Solids-Diff}. In this case, \eqref{Eq:Solids-Diff} is written as 
\begin{equation} \label{Eq:Routh-Diff} 
\left(\!\!\begin{array}{c} f' \\ g'\end{array}\!\!\right) =  \left(\!\!\begin{array}{cc}\tau_1 & -1 \\ 1 & 0  \end{array}\!\!\right)[\mathcal{QP}]^T  \left(\!\!\begin{array}{c} f\\ g\end{array}\!\!\right) \ \ \mbox{for} \ \  [\mathcal{QP}]^T \!= \frac{mr}{P(\gamma_3)} \left(\!\!\begin{array}{cc} -r(\mathbb{I}_3 -r\sigma)   &  l(\mathbb{I}_3 -r\sigma) \\   -r(\mathbb{I}_1\gamma_3 +\zeta\sigma)   & l (\mathbb{I}_1\gamma_3 +\zeta\sigma) 
\end{array}\!\!\right).
\end{equation}

Observe that the kernel of $[\mathcal{QP}]^T$ is generated by the constant vectors $(l,r)$. That is,  the constant functions $f_1(\gamma_3) = l$ and $g_1(\gamma_3) = r$ are solutions of the differential equations \eqref{Eq:Routh-Diff}. Therefore we obtain the first horizontal gauge momentum $J_1$ given by $J_1 = lj_1+rj_2 = - \langle {\bf M}, s\rangle$.
On the other hand, we see that  the functions
$$
f_2(\gamma_3) = \frac{-\mathbb{I}_1 - ml\zeta}{\sqrt{P(\gamma_3)}} \qquad \mbox{and} \qquad g_2(\gamma_3) = \frac{ - mr\zeta}{\sqrt{P(\gamma_3)}}
$$
satisfy that $[\mathcal{QP}]^T \left(\!\!\begin{array}{c} f_2 \\ g_2 \end{array}\!\!\right) = [\mathcal{QP}]^T \left(\begin{array}{c} \frac{-\mathbb{I}_1}{\sqrt{P(\gamma_3)}}\\ 0 \end{array}\right). $
Now, we can directly check that 
\begin{equation} \label{Eq:Routh-Diff3} 
\tfrac{-m\mathbb{I}_1}{P(\gamma_3)\sqrt{P(\gamma_3)}}\left(\begin{array}{cc}\tau_1 & -1 \\ 1 & 0  \end{array}\right)  [\mathcal{QP}]^T \left(\begin{array}{c} 1 \\ 0 \end{array}\right) =    \left(\!\!\begin{array}{c} f'_2 \\ g'_2 \end{array}\!\!\right).  
\end{equation}

Therefore, $J_2 = f_2j_1+g_2j_2 = \sqrt{P(\gamma_3)} \Omega_3$ is a horizontal gauge momentum with associated horizontal gauge symmetry $\chi_2 = f_2\xi_1+g_2\xi_2$.  The functions $J_1$ and $J_2$ were found in \cite{routh}, see also \cite{Cushman1998}.

Since $J_1$ and $J_2$ are $G$-invariant functions on $\M$, from Theorem \ref{T:Solids-Casimirs}, they induce the Casimir functions $\bar{J}_1$ and $\bar{J}_2$ of the reduced bracket $\{\cdot, \cdot\}_\red^\B$, and following Corollary \ref{C:Solids-Poisson} we obtain that $\{\cdot, \cdot\}_\red^\B$ is Poisson on $\bar\M$.   As we already pointed out in the general case, away from the singularies of $\bar\M$, the reduced bracket $\{\cdot, \cdot\}_\red^\B$, describing the (reduced) dynamics of the Routh sphere, has rank 2 and has a foliation given by 2-dimensional symplectic leaves.

\subsection{The rolling ellipsoid}

\subsubsection{Preliminaries and hamiltonization}

Consider the geometrically axisymmetric ellipsoid
$$
\frac{x^2}{b}+ \frac{y^2}{b} + \frac{z^2}{c} = 1, 
$$ 
with its center of mass coinciding with its geometric center.  We assume also that, after choosing a moving frame whose axes coincide with the principal axes of inertia of the ellipsoid, the inertia tensor  has the form $\mathbb{I}= diag(\mathbb{I}_1, \mathbb{I}_1, \mathbb{I}_3)$.  
In this case, 
$$
s = -\frac{(b\gamma_1 , b\gamma_2,c \gamma_3)}{\sqrt{b(1-\gamma_3^2) + c\gamma_3^2}} = \frac{-\mathbb{B}\vecgamma }{\sqrt{\langle \mathbb{B}\vecgamma,\vecgamma\rangle}},
$$
for $\mathbb{B}= diag(b,b,c)$. That is, 
$$
\varrho(\gamma_3) =  \frac{-b}{\sqrt{b(1-\gamma_3^2) + c\gamma_3^2}} = \frac{-b}{\sqrt{\langle \mathbb{B}\vecgamma,\vecgamma\rangle}} \qquad \mbox{and}  \qquad \zeta(\gamma_3)=  \frac{- c \gamma_3}{\sqrt{b(1-\gamma_3^2) + c\gamma_3^2}} = \frac{-c\gamma_3}{\sqrt{\langle \mathbb{B}\vecgamma,\vecgamma\rangle}}
$$

Therefore, following \eqref{Eq:Solids-JK}, we see that
\begin{equation*}
 %\begin{split}
  \langle \mathcal{J},\mathcal{K}_\subW\rangle  =  mb\langle \vecOm, d\vecL\rangle + \mathcal{Q}\langle \vecgamma, \vecL\rangle + \mathcal{P}d\lambda_3,
%   - \tfrac{ mb}{\sqrt{\langle\mathbb{B}\vecgamma,\vecgamma\rangle}}\langle \vecOm, \vecgamma\rangle  \langle s, d\vecL\rangle 
%  - \tfrac{mb(c-b)}{\langle \mathbb{B}\vecgamma, \vecgamma\rangle^{3/2} } \gamma_3 c_3 \langle \vecgamma , d\vecL\rangle - \tfrac{m(b-c)}{\sqrt{\langle \mathbb{B}\vecgamma, \vecgamma\rangle}}  \left[ \tfrac{\gamma_3^3c_3 }{\langle \mathbb{B}\vecgamma, \vecgamma\rangle} + 1 \right] d\lambda_3,
% \end{split}
\end{equation*}
where, using \eqref{Eq:Q-P}, we compute
\begin{equation*}
 \begin{split}
  \mathcal{Q} =    - \tfrac{ mb}{\langle\mathbb{B}\vecgamma,\vecgamma\rangle}\left[b \langle \vecOm, \vecgamma\rangle + \tfrac{b(b-c)}{\sqrt{\langle\mathbb{B}\vecgamma,\vecgamma\rangle}} \gamma_3 c_3\right] \quad  \mbox{and} \quad
  \mathcal{P} =  \tfrac{b(c-b)}{\sqrt{\langle\mathbb{B}\vecgamma,\vecgamma\rangle}} \left[ b \tfrac{\langle \vecOm, \vecgamma \rangle}{\sqrt{\langle\mathbb{B}\vecgamma,\vecgamma\rangle}} \gamma_3 - \tfrac{(b-c)}{\langle\mathbb{B}\vecgamma,\vecgamma\rangle} \gamma_3^2 c_3 +1 \right].
 \end{split}
\end{equation*}

From \eqref{Eq:DynGauge}, we obtain the dynamical gauge transformation defined by
$$
B=  mb\langle \vecOm, d\vecL\rangle .
$$

Following Theorem \ref{T:Solids-Poisson} we conclude that the reduced dynamics of the ellipsoid is described by the Poisson bracket $\{\cdot, \cdot \}_\red^\B$ on $\bar\M$ induced by the gauge related bracket $\{\cdot, \cdot\}_\B$ on $\M$.

\subsubsection{Conserved quantities and Casimirs}

Following Section \ref{Ss:Casimir} we see that the gauge momenta associated to the rolling ellipsoid are the functions
$$
J_1 = f_1(\gamma_3)  j_1 + g_1(\gamma_2)  j_2  \quad \mbox{and} \quad J_2 = f_2(\gamma_3)  j_1 + g_2(\gamma_2)  j_2,
$$
where $f_i$ and $g_i \in C^\infty(Q)^G$ satisfy \eqref{Eq:Solids-Diff}.

In our case, the matrix $[\mathcal{QP}]$ has the form 
\begin{equation} \label{Eq:Ellipsoid-QP}
[\mathcal{QP}] = \frac{m}{P(\gamma_3)} \left( [\mathcal{QP}]_1 +\frac{mb}{\langle {\mathcal B}\vecgamma,\vecgamma\rangle}\left(\begin{array}{cc} b^2 & bc \\ b(b-c) & c(b-c)\gamma_3^2 \end{array}\right)  \right),
\end{equation}
where 
$$
[\mathcal{QP}]_1 = \frac{1}{\langle \mathcal{B}\vecgamma, \vecgamma\rangle} \left( \begin{array}{cc} -b(b+ \tfrac{(c-b)}{\langle \mathcal{B}\vecgamma, \vecgamma\rangle}\gamma_3^2) 
& -b^2\gamma_3 (-1 +\tfrac{(c-b)}{\langle\mathcal{B}\vecgamma, \vecgamma\rangle}(1-\gamma_3^2) ) \\
(c-b)^2 \gamma_3 (1- \tfrac{\gamma_3}{\langle \mathcal{B}\vecgamma, \vecgamma\rangle} )  
& b(b-c) (\tfrac{(c-b)}{\langle \mathcal{B}\vecgamma, \vecgamma\rangle} \gamma_3(1-\gamma_3^2) -1 )
\end{array} \right) \left(\begin{array}{cc} \mathbb{I}_3 & 0 \\ 0 & \mathbb{I}_1 \end{array}\right).
$$

The linear system \eqref{Eq:Solids-Diff} of ordinary differential equations has two $G$-invariant solutions $J_1, J_2$, and we see that these conserved quantities can be seen as a linear combination of the function $j_1$ and $j_2$. 
Since the 2-form $B$ was chosen in such a way that the induced functions $\bar{J}_1, \bar{J}_2$ on $\bar\M$ are Casimirs of the reduced bracket $\{\cdot, \cdot\}_\red^\B$ (see Theorem \ref{T:Solids-Casimirs}) we obtain that $\{\cdot, \cdot\}_\red^\B$ on $\bar\M$ is Poisson as a consequence of Corollary \ref{C:Solids-Poisson}.  

Moreover, following Section \ref{Ss:Regular} we conclude that away from the singularities (where $\bar\M$ is a 4-dimensional manifold), the  reduced bracket $\{\cdot, \cdot\}_\red^\B$ has symplectic leaves of dimension 2.

\begin{remark}\label{R:Bolsinov}
In \cite{Bolsinov} it was raised out the question of whether the system could be described by a Poisson bracket  on $\M/E(2)$.  From \eqref{Eq:Q-P},  we see that $\langle \mathcal{J},\mathcal{K}_\subW\rangle - B = {\mathcal Q}\, \langle \vecgamma, d\vecL\rangle + {\mathcal P}\, d\lambda_3$.  As in the case of the Routh sphere,  the term ${\mathcal Q}\, \langle \vecgamma, d\vecL\rangle$ is a well defined 2-form on $\M/E(2)$ while $ {\mathcal P}\, d\lambda_3$ is not.  The fact that $\mathcal{P} \neq 0$ is the reason why the reduced bracket on $\M/E(2)$ induced by $\{\cdot, \cdot\}_\B$ is not twisted Poisson, see item $(iii)$ in Sec.~\ref{Ss:Regular}.
Moreover, considering $(\vecgamma, {\bf M})$ the coordinates on $\M/E(2)$, we can check also that the (partially) reduced bracket $\{\cdot , \cdot \}_{1}^\B$ on $\M/E(2)$ induced by $\{\cdot , \cdot \}_\B$ and the orbit projection $\rho_1:\M\to \M/E(2)$ fails to satisfy the Jacobi identity. In fact, using \eqref{Eq:Jac-pi_Bred} we can check that  $d(\langle \mathcal{J},\mathcal{K}_\subW\rangle - B) (\pi_\B^\sharp(d\rho_1^*(dM_1)),\pi_\B^\sharp(d\rho_1^*(dM_2)),\pi_\B^\sharp(d\rho_1^*(d\gamma_1))) \neq 0$. 

We could try to find another dynamical gauge transformation. However, any other choice of a 2-form $B$ involves a term $\langle \vecgamma,d\vecL\rangle$ or $d\lambda_3$, which are not ``dynamical'' in the sense that ${\bf i}_{X_\nh}\langle \vecgamma,d\vecL\rangle = \langle \vecgamma, \vecOm \times \vecL\rangle \neq 0$ and ${\bf i}_{X_\nh}d\lambda_3 = \Omega_1\gamma_2-\Omega_2\gamma_1 \neq 0$.

Another view-point, in connection with the observations made in \cite{Bolsinov}, is that none of the conserved quantities are gauge momenta with respect to the $E(2)$-action.  In other words,  none of the functions $\bar{J}_i$ are  Casimirs of the (partially) reduced bracket $\{\cdot , \cdot \}_{1}^\B$ on $\M/E(2)$ as we saw in Remark~\ref{R:Solids-not-Poisson}.
\end{remark}

\begin{remark}\label{R:Chaplygin}
 If $c=b$ then we are describing a inhomogeneous ball of radio $b$ with the center of mass in the geometric center, which is the {\it Chaplygin ball}, Example \ref{Ex:ChapBall}.  In this case, $\langle\mathbb{B}\vecgamma,\vecgamma\rangle = 1$ and $s=-b\vecgamma$. Hence we have that 
 $$
 \langle \mathcal{J},\mathcal{K}_\subW\rangle  =  mb\langle \vecOm, d\vecL\rangle 
  + mb^2\langle \vecOm, \vecgamma\rangle  \langle \vecgamma, d\vecL\rangle 
  $$
 and $ \langle \mathcal{J},\mathcal{K}_\subW\rangle - B =  mb^2\langle \vecOm, \vecgamma\rangle  \langle \vecgamma, d\vecL\rangle $, which is a well defined 2-form on $\M/E(2)$ (observe that here ${\mathcal P}=0$). That is why in the reduced manifold $\M/E(2)$ we have a twisted Poisson bracket describing the dynamics, see Ex.\ref{Ex:ChapBall} and \cite{Jac} for more details.  Therefore the $S^1$-reduction is not needed and that is why we can avoid the hypothesis of the axisymmetric distribution of mass for the Chaplygin ball. 
\end{remark}


\begin{thebibliography}{99}

\begin{small}

\bibitem{Jac} P. Balseiro; 
The Jacobiator of nonholonomic systems and the geometry of reduced nonholonomic brackets. 
{\it Arch. Ration. Mech. Anal.} {\bf 214}, (2014), 453--501.

\bibitem{PO2015} P. Balseiro, O.E. Fernandez; 
  Reduction of nonholonomic systems in two stages and Hamiltonization.
  {\it Nonlinearity}, {\bf 28}, (2015), 2873.
  
\bibitem{PL2011} P. Balseiro, L. Garcia-Naranjo:  Gauge Transformations, Twisted Poisson
Brackets and Hamiltonization of Nonholonomic Systems. {\it Arch. Ration. Mech. Anal.}, {\bf 205}(1),
 (2012), 267-310.
 
 \bibitem{BalSan} P. Balseiro, N. Sansonetto: A geometric characterization of certain first integrals for nonholonomic systems with symmetries. \emph{SIGMA}12 018, 14 pages, (2016).
 
 \bibitem{Yapu} P. Balseiro, L. P. Yapu-Quispe, work in progress.

\bibitem{Book:BC} L.Bates, R. Cushman; {\it Global Aspects of Classical Integrable Systems}, Second Edition, Birkhauser. 

\bibitem{BGMConservation} L. Bates, H. Graumann, C. MacDonell; Examples of gauge conservation laws in nonholonomic systems.
{\it Rep. Math. Phys.} \textbf{37},  (1996), 295--308.

\bibitem{BS2016} L. Bates, J. Sniatycki; Not quite hamiltonian reduction. {\it Rep. Math. Phys.}, 
{\bf 76}(1), (2015), 41–52.

\bibitem{BS93} L. Bates, J. Sniatycki; Nonholonomic reduction.  \emph{ Rep. Math. Phys.}
\textbf{32}(1),  (1993), 99--115.

\bibitem{Bizy-Tsiganov} I.~A.~Bizyaev, A.~V.~Tsiganov; On the Routh sphere problem. \emph{Journal of Physics A: Mathematical and Theoretical}, vol. 46, no. 8, 085202 (2013) 11 pp.

\bibitem{Blochbook} A.~M. Bloch; \textit{Nonholonomic mechanics and control}. Springer Verlag, New York, 2003

\bibitem{BKMM} A.~M. Bloch, P.~S.Krishnaprasad , J.~E. Marsden, R.~M. Murray; Nonholonomic mechanical systems with symmetry.
{\it Arch. Ration. Mech. Anal.}, {\bf 136}, (1996), 21--99.

\bibitem{Bolsinov} A. V. Bolsinov, A. A. Kilin, A. O. Kazakov; Topological monodromy as an obstruction to
Hamiltonization of nonholonomic systems: pro or contra?.  {\it Journal of Geometry and Physics}, {\bf 87}, (2014),  61-75.

\bibitem{BorisovMamaev2002} A.~V.~Borisov, I.~S.~Mamaev; Rolling of a rigid body on plane and sphere. Hierarchy of dynamics. {\emph Regul. Chaotic Dyn.}, (2002),  277–328. 

\bibitem{BorisovMamaev} A.~V. Borisov,  I.~S. Mamaev; Chaplygin's ball rolling problem is Hamiltonian. \textit{ Math. Notes} {\bf 70}, (2001), 793--795.


\bibitem{BorisovMamaev2008} A.~V. Borisov, I.~S. Mamaev; Conservation laws, hierarchy of dynamics and explicit
 integration of nonholonomic systems.  \textit{Regul. Chaotic Dyn.}, {\bf 13}, (2008), 443--490.
 
\bibitem{BoMamTsig2014} A.~V. Borisov, I.~S. Mamaev,  A.~V.~Tsiganov; Non-holonomic dynamics and Poisson geometry. {\emph Russian Mathematical Surveys},  {\bf 69}(3), (2014)  481 - 538.

\bibitem{Chapligyn_reducing_multiplier} S.~A. Chaplygin;  On the theory of the motion of nonholonomic systems. The reducing-multiplier Theorem.  \textit{ Regul.  Chaotic Dyn.}, {\bf 13},
 369--376 (2008); Translated from \textit{ Matematicheski\u{i}  sbornik} (Russian)  {\bf 28} (1911), by A. V. Getling.  %no. 4,


\bibitem{Crampin-Mestdag} M. Crampin, T. Mestdag; The Cartan Form for Constrained Lagrangian Systems and the Nonholonomic Noether Theorem. {\it Int. J. Geom. Methods Mod. Phys.}, {\bf 8 }, (2011), 897--923.


\bibitem{CdLMD} F. Cantrijn, M. De Le\'on, D. Mart\'\i n de Diego; On almost-Poisson structures in nonholonomic mechanics.
\emph{Nonlinearity}, {\bf 12} (3), (1999), 721.

\bibitem{Cushman1998} R.Cushman; Routh's sphere. {\it  Rep. Math. Phys.}, {\bf 42} No. 1--2, (1998), 47-70.

\bibitem{Book:CDS} R. Cushman, H. Duistermaat, J. Sniatycki; {\it Geometry of Nonholonomically Constrained Systems}, Advance Series in Nonlinear Dynamics, Vol. 26, World Scientific (2010).

\bibitem{Duistermaat} H.~Duistermaat; Chapligyn’s Sphere. {\it arXiv:math.DS/0409019} (2000)

\bibitem{EhlersKoiller}  K. Ehlers, J. Koiller, R. Montgomery, P.~M. Rios;  Nonholonomic  systems via moving frames: Cartan
equivalence and Chaplygin Hamiltonization. \textit{ The breath of Symplectic and Poisson Geometry, Progr. Math.}
{\bf 232}, \textit{ Birkh\"auser Boston, Boston MA,},  (2005), 75--120.




\bibitem{FSG2008} F. Fass\`o, A. Giacobbe, N. Sansonetto; Gauge conservation laws and the momentum equation in nonholonomic mechanics.  {\it  Rep. Math. Phys.}, {\bf 62} (2008), 345--367.

\bibitem{FSG2009} F. Fass\`o, A. Giacobbe, N. Sansonetto; On the number of weakly Noetherian constants of motions of nonholonomic systems. {\it Journal of Geometric Mechanics}, {\bf 1} (2009), 389--416.


\bibitem{FS2015} F. Fass\`o. and N. Sansonetto; Conservation of `Moving' Energy in Nonholonomic Systems with Affine Constraints and Integrability of Spheres on Rotating Surfaces. {\it J. Nonlinear Sci.} (2015)

\bibitem{FedorovJovan} Yu.~N. Fedorov, B. Jovanovi{\'c}; Nonholonomic LR systems as generalized Chaplygin systems with an invariant measure and flows on homogeneous spaces. \textit{J. Nonlinear Sci.} \textbf{14}, (2004), 341--381.


\bibitem{FedorovKozlov}  Yu.~N. Fedorov  and V.~V.  Kozlov;  Various aspects of $n$-dimensional rigid body dynamics. \textit{ Dynamical systems in classical mechanics, Amer. Math. Soc. Transl. Series 2}, \textbf{168}, \textit{Amer. Math. Soc. Providence, RI,} (1995), 141--171. 


\bibitem{Fernandez}  O. Fernandez, T. Mestdag, A.~M. Bloch; A generalization of Chaplygin's reducibility
Theorem. \textit{Regul. Chaotic Dyn.} \textbf{14},  (2009),  635--655. %no. 6,


\bibitem{Naranjo2008} L.~C. Garc\'ia-Naranjo; Reduction of  almost Poisson brackets and Hamiltonization of the Chaplygin sphere. \textit{ Disc. and Cont. Dyn. Syst. Series S} {\bf 3}, (2010), 37--60.
 
 
 \bibitem{LuisJames} L.~C. Garc\'ia-Naranjo,  James Montaldi; Gauge momenta as Casimir functions of nonholonomic systems. {\emph Preprint}, (2016).


\bibitem{IbLeMaMa1999} A. Ibort, M. de Le\'on, J.~C. Marrero, D. Mart\'in de Diego;  Dirac brackets in constrained dynamics. \textit{ Fortschr. Phys.} {\bf 47}(5), (1999), 459--492.  


\bibitem{JovaChap}  B. Jovanovi{\'c}; Hamiltonization and integrability of the Chaplygin sphere in $\mathbb{R}^n$.  \textit{J. Nonlinear Sci.}  {\bf 20}(5), (2010),  569-593.  %no. 5,

\bibitem{Jovanovic} B. Jovanovi{\'c};  Symmetries of line bundles and Noether theorem for time-dependent nonholonomic
systems. Preprint.  

\bibitem{Marle95} C.~M. Marle; Reduction of constrained mechanical systems and stability of relative equilibria.
{\it Comm. Math. Phys.} {\bf 174}, (1995), 295--318. 


\bibitem{Marle} C. M. Marle; Various approaches to conservative and nonconservative nonholonomic systems, \emph{Rep. Math. Phys.}, 42, (1998), 211–229.

\bibitem{MarsdenKoon} J.~E.~Marsden, W.~S.~Koon; The Poisson reduction of nonholonomic mechanical systems,
{\it  Rep. Math. Phys.}, 42, (1998), 101-134.

\bibitem{Book:MR} J.E.~Marsden, T.~Ratiu; {\it Introduction to mechanics and symmetry}, Volume 17, Texts in Applied Mathematics, Second Edition. Springer-Verlag (2002).


\bibitem{Ohsawa} T. Ohsawa, O. Fernandez, A.~M. Bloch, D. Zenkov; Nonholonomic Hamilton-Jacobi theory via Chaplygin Hamiltonization. \textit{ J. Geom. Phys.} {\bf 61}, (2011), 1263--1291.


\bibitem{Ramos2004} A.~Ramos; Poisson structures for reduced non-holonomic systems, {\emph J. Phys. A: Math. Gen.} Vol.37 (17), (2004), 4821.

\bibitem{routh}  E. Routh;  {\emph A treatise on the dynamics of a system of rigid bodies}, Parts I and II, Dover, New York (1960).

\bibitem{SW} P. \v{S}evera, A. Weinstein; Poisson Geometry with a 3-form Background, {\emph Progr. Theoret.
Phys.}, {\bf 144}, (2001), 145--154.


\bibitem{sniatycki} J. Sniatycki; Nonholonomic Noether Theorem and Reduction of Symmetries.
{\it Rep. Math. Phys.}, {\bf 42},  (1998), 5--23.

\bibitem{SchaftMaschke1994} A. J. van der Schaft,  B. M. Maschke;  On the Hamiltonian
formulation of nonholonomic mechanical systems. {\it Rep. on Math. Phys.}, {\bf 34}, (1994), 225--233.

\bibitem{Zenkov} D.V. Zenkov; Linear conservation laws of nonholonomic systems with symmetry.
{\it Disc. Contin. Dyn. Syst.}, (2003), 967--976.


\end{small}

\end{thebibliography}
\end{document}